\documentclass{article}

\RequirePackage[numbers]{natbib}
\usepackage{xcolor,pgf}
\usepackage{natbib}

\usepackage{graphicx}
\usepackage{float}
\usepackage{placeins}
\usepackage[english]{babel}
\usepackage{amsmath}
\usepackage{arydshln}
\usepackage{float}
\usepackage{url}
\usepackage{amsmath, amsthm, amssymb}
\usepackage{transparent}
\usepackage{enumerate,amsxtra,comment}
\usepackage{hyperref}
\usepackage{dsfont}
\usepackage{animate}
\usepackage{graphicx}
\allowdisplaybreaks

\newcommand{\argmin}{{\operatorname{argmin}}}

\newcommand{\R}{\mathds{R}}
\newcommand{\C}{\mathds{C}}

\newcommand{\N}{\mathds{N}}
\newcommand{\Z}{\mathds{Z}}
\newcommand{\ind}{\mathds{1}}

\newcommand{\Yr}{Y^{(r)}}
\newcommand{\fraco}{\frac{1}}

  \newcommand{\citeasnoun}{\citet}

  \newcommand{\Ad}{\mathbf{Ad}}

  \newtheorem{thm}{{Theorem}}[section]
  \newtheorem{lem}[thm]{{Lemma}}
  
  \newtheorem{asp}{{Assumption}}
  \newtheorem{df}[thm]{{Definition}}

    \newcommand{\Mm}{\odot}
    \newcommand{\Xb}{\mathbf{X}}

        \newcommand{\veco}{\operatorname{vec}}

        \def\eps{\varepsilon}

        \newcommand{\var}{\mbox{Var}}
        \newcommand{\sign}{\operatorname{sign}}
        \newcommand{\cov}{\mbox{Cov}}

\begin{document}

{\bf \center \huge Time Series Modeling on Dynamic Networks}\\
{\bf J. Krampe$^{1}$} \\ \newline

\noindent $^1$ University of Mannheim and TU Braunschweig \\

\begin{abstract}
This paper focuses on modeling the dynamic attributes of a dynamic network with a fixed number of vertices. These attributes are considered as time series which dependency structure is influenced by the underlying network. They are modeled by a multivariate doubly stochastic time series framework, that is we assume linear processes for which the coefficient matrices are stochastic processes themselves. We explicitly allow for dependence in the dynamics of the coefficient matrices as well as between these two stochastic processes. This framework allows for a separate modeling of the attributes and the underlying network. In this setting, we define network autoregressive models and discuss their stationarity conditions. Furthermore,  an estimation approach  is discussed in a  low- and high-dimensional setting  and how this can be applied to forecasting. The finite sample behavior of the forecast approach is investigated. This approach is applied to real data whereby the goal is to forecast the GDP of $33$ economies.
\end{abstract}

\section{Introduction}

Consider a vertex-labeled dynamic and weighted network with a fixed number $d$ of vertices given by the set $V=\{1,\dots,d\}$. The weights are within the interval $[-1,1]$.
Such a dynamic and weighted network with a fixed number of vertices can be described by a time dependent adjacency matrix, here denoted by $\Ad=\{Ad_t, t \in \Z\}$, where $Ad_t$ is $[-1,1]^{d \times d}$- valued. If no edge is present at time $t$, a zero weight is considered. Thus, $e_i ^\top Ad_{t} e_j$ gives the weight of an edge at time $t$ from vertex $i$ to vertex $j$. It is considered that the network is driven by some random process, hence, the corresponding adjacency matrix process  $\Ad$ is a stochastic process.
 
The vertices are considered as actors (e.g. people, countries), and a network could describe a relationship structure among these actors. A (weighted) edge between two actors describes some connection between them. The weights can be interpreted as the strength of the connection. For example, consider economies as actors where a possible relationship between two economies could be given by their relative trade volume at a given time point. Further examples are social networks, see \cite{blei2007statistical,morris1997concurrent}. The actors in such networks often possess attributes. These attributes can be static (e.g a person's name or birthday) or dynamic (e.g. personal income, time a person does sports, or political views). In the example with the economies and their trade volume as a relationship between them, dynamic attributes of interest are macroeconomic measures such as inflation rate or gross domestic product (GDP). These attributes may be affected by the attributes of other actors, especially by actors with which the considered actor is strongly connected. In this work, the dynamic attributes are denoted by a time series $\Xb=\{ X_{t}, t \in \Z\}$. To simplify the notation, we focus on the case that each actor has only one attribute, meaning that $\Xb$ is d-dimensional. Nevertheless, this framework can also handle multiple attributes per actor, see Section \ref{sec.3.2} for details.

In this work, the dynamic attributes are denoted by a $d$-dimensional time series $\Xb=\{ X_{t}, t \in \Z\}$, where each component of the time series is assigned to a vertex (actor) of the underlying network. In the social-economical literature, the influence of connected actors on the attributes is denoted as peer effect, see \cite{goldsmith2013social,manski1993identification}.

This work focuses on the dynamic attributes and not on the network itself. Consequently, this work is not about modeling a dynamic network. For modeling these dynamic networks, many models for static networks have been extended to the dynamic case as done by \citeasnoun{hanneke2010discrete,krivitsky2014separable} for the Exponential Random Graph Models (ERGM), see Section 6.5 in \cite{Kolaczyk:2009}, or by \citeasnoun{xing2010state,xu2015stochastic} for the stochastic block model (SBM). By contrast, this work gives a framework which models the  dynamic attributes, that means modeling a time series on a dynamic network, in which the weighted edges influence the dependency of the time series. \citeasnoun{knight2016modelling, zhu2017network} have modeled these attributes for non-random edges, which mainly cover static networks. In the context of a static network,  attributes can be considered as \emph{standard} multivariate time series with additional information and can be modeled by using vector autoregressive (VAR) models with constraints, see \citet[Chapter 2 and 5]{luetkepohl2007new}. However, VAR models have many parameters, which is why \citeasnoun{knight2016modelling, zhu2017network} focus on how to use the network structure to reduce the number of parameters so that high dimensions, meaning a large number of vertices, become feasible. By contrast, this work deals with a random network structure, and consequently, the process $\Xb$ cannot be modeled appropriately by using VAR models. Instead, we adopt a multivariate doubly stochastic time series framework, meaning two stochastic processes drive the time series. On the one hand, there is the innovation process of the time series. On the other hand, the coefficient matrices in linear processes or autoregressive models are stochastic processes themselves. Doubly stochastic time series models were introduced in \citet{tjostheim1986some,pourahmadi1986stationarity,pourahmadi1988stationarity}, and these authors assume that the two processes driving the time series are independent. That would mean the network could influence the attributes but not the other way around. This assumption could be too restrictive for most application. In the example where the GDP is the attribute and the trade volume defines the underlying network, the influence goes in both directions. To capture such behavior, we allow dependency between both processes, see Section~\ref{sec.3.2} for details. Hence, the network, in form of the edges, can influence the attributes, and the attributes can influence the edges. 

\citet{knight2019generalised} extended the model of \cite{knight2016modelling} to dynamic networks. Their model can be considered as a special case of the model considered in this paper, see Section~\ref{sec.3.2} for details. They present an \emph{R}-package to estimate their model, however, they present theoretical results only for the case of a static network such that the attributes can be written as a VAR model with constraints. Thus, they do not derive theoretical results for the dynamic case. This gap is filled in this work.


This paper is structured as follows. In Section \ref{sec.3.2}, time series on dynamic networks are defined and some basic properties are given. In Section \ref{stat.result}, the focus is on statistical results of network autoregressive processes, and their applications to forecasting are discussed. Some of the forecasting results are underlined by a simulation study which is given in Section \ref{Simulation}. In Section \ref{net.real.data}, we apply this setup to forecast the GDP using the trade volume between economies as an underlying network. Proofs can be found in Section \ref{net.proofs}.

\section{Time Series Modeling on Dynamic Networks} \label{sec.3.2}
To elaborate, we first  fix some notation. For  a random variable $X$, we write   $ \|X\|_{E,q}$ for $\big(E|X|^q\big)^{1/q}$, where   $q\in \N$;  for  a vector $x\in \R^d$,  $\|x\|_0 = \sum_{j=1}^d \ind(x_j \not = 0)$, $ \| x \|_1 = \sum_{j=1}^d |x_j|$ and $\| x \|_2^2 = \sum_{j=1}^d |x_j|^2$.  Furthermore, for a $r\times s$ matrix $B=(b_{i,j})_{i=1,\ldots,r, j=1,\ldots,s}$,  $\|B\|_1=\max_{1\leq j\leq s}\sum_{i=1}^r|b_{i,j}|=\max_j \| B e_j\|_1$, 
$\|B\|_\infty=\max_{1\leq i\leq r}\sum_{j=1}^s|b_{i,j}|=\max_{i} \| e_i^\top B\|_1$, $\|B\|_{\max}=\max_{i,j} |e_i^\top B e_j|$, where $e_j=(0,\ldots,0,1,0,\ldots, 0)^\top$ denotes  the  vector with the one appearing in the $j$th position. Let $\ind=(1,\dots,1)^\top$ be a vector of ones. For a matrix $B$, the absolute value evaluated component-wise is denoted by $|B|$. Denote the largest eigenvalue of a matrix $B$ by $\rho(B)$ and $\|B\|_2^2=\rho(BB^\top)$. The $d$-dimensional identity matrix is denoted by $I_d$. Furthermore, for two matrices $A,B$ the Kronecker product of $A$ and $B$ is denoted by $A\otimes B$, see among others Appendix A.11 in \cite{luetkepohl2007new}. Let $A\Mm B$ denote the component-wise multiplication of $A$ and $B$, i.e. the Hadamard product.  $\operatorname{sign}(\cdot)$ denotes the  signum function, $|\cdot|^+=\max(0,\cdot)$, and they are evaluated component-wise for matrix arguments. Let $I_{d;-I}\in \R^{(d-|I|)\times d}$ denote a $d$-dimensional identity matrix without the rows $i \in I$ and $I_{d;I}=I_{d;-I^C}$. An empty product denotes the neutral element, meaning $\prod_{k=1}^0 B_k=I_d$. For a vector-valued times series $\{X_t\}$, we write $X_{t;r}:=e_r^\top X_t$, and for a matrix-valued time series $\{Ad_t\}$, we write $Ad_{t;rs}:=e_r^\top Ad_t e_s$.

With this, we can define a network linear process as follows.
\begin{df} \label{def.GNLP}
Let $\Ad=\{Ad_t, t \in \Z\}$ be a $[-1,1]^{d \times d}$-valued, strictly stationary stochastic process, and let $f_j : \R^{(d\times d)j} \to \R^{d\times d}$ be measurable functions. Furthermore, let $\eps=\{\eps_t,t \in \Z\}$ be an i.i.d. sequence of $\R^d$-valued random vectors with $E\eps_{1}=\mu \in \R^d, \var(\eps_{1}) =\Sigma_\eps$ (positive definite and $\|\Sigma_\eps\|_2 <\infty$).  $\{\eps_s, s> t\}$ and $\{Ad_s, s\leq t\}$ are independent for all $t$. If the following $L_2$-limes exists,
\begin{align}
X_{t}= \sum_{j=1}^\infty f_j(Ad_{t-1},\dots,Ad_{t-j}) \eps_{t-j} + \eps_{t}=: \sum_{j=1}^\infty B_{t,j} \eps_{t-j} + \eps_{t}, \label{def.NLP}
\end{align}
we denote the process given by $\Xb=\{X_{t}, t\in \Z\}$ a (generalized) network linear process (GNLP). 

Let $p, q \in \N$ and  $f_j : \R^{(d\times d)j} \to \R^{d\times d}, g_s :  \R^{(d\times d)s} \to \R^{d\times d}, j=1,\dots,p, s=1,\dots,q$ be measurable functions. A process $\Xb$ fulfilling equation $(\ref{def.narma})$ is denoted as a (generalized) network autoregressive moving average process of order $(p,q)$ (GNARMA$(p,q)$)
\begin{align}
X_{t} = \sum_{j=1}^p f_j(Ad_{t-1},\dots,Ad_{t-j})X_{t-j} + \sum_{s=1}^q g_s(Ad_{t-1},\dots,Ad_{t-s}) \eps_{t-s} + \eps_{t}. \label{def.narma}
\end{align}
\end{df}

In this work, the focus is on the following network autoregressive process of order $p$ given by 
\begin{align}\label{GNAR.eq}
X_t=\sum_{j=1}^p (A_j \Mm G_j(Ad_{t-j})) X_{t-j}+\eps_{t},
\end{align}
where $A_1,\dots,A_p \in \R^{p \times p}$ are coefficient matrices, $G_j : [-1,1]^{d\times d}\to [-1,1]^{d\times d},j=1,\dots,p$ are some known measurable functions.  Since \eqref{GNAR.eq} is a special case of \eqref{def.narma}, where the functions $f_j$ appearing in \eqref{def.narma} have a particular form, we drop the term generalized. Note that the causal solution of \eqref{GNAR.eq} fits into the framework \eqref{def.NLP}, see Lemma~\ref{lem.solution}. Since for an  adjacency matrix $Ad\in\{0,1\}^{d\times d}$ we have that $e_r Ad^j e_s$ gives the number of paths with length $j$ from node $r$ to node $s$, examples for $G_j$ are polynomials.   E.g., if for some lag $j$ the direct neighbors as well as the neighbors of these neighbors should have a direct impact, then $G_j$ can be chosen as  $G_j(Ad_{t_j})=\operatorname{sign}(Ad_{t_j}+Ad_{t_j}^2).$ Denote a vertex $v$ as a $k$-stage neighbor of $u$ if there is a path from $u$ to $v$ of length $k$ but no shorter one. That means a direct neighbor is a $1$-stage neighbor and a neighbor's neighbor which is not a direct neighbor a $2$-stage one. Given an adjacency matrix $Ad\in \{0,1\}^{d\times d}$, the $k$-stage neighborhood matrix  is given by $\mathcal{N}_k(\Ad)=\sign(|\sign((Ad^\top)^k)-\sign(\sum_{i=1}^{k-1} (Ad^\top)^i)|^+)$, where the ones in $e_j^\top \mathcal{N}_k(\Ad)$ indicate the $k$-stage neighbors of vertex $j$.

For direct edges, two natural concepts occur; the concept that the influence goes in the direction of the edge and vice versa. Definition~\ref{def.NLP} can handle both concepts.  E.g., if $\tilde G_j(\cdot)=G_j(\cdot)^\top,j=1,\dots,p,$ is used in \eqref{GNAR.eq},  one can switch between both concepts. If not specified otherwise,  the concept that the influence goes in the direction of the edge is used in this work. That means the easiest function for $G_j$ in  model \eqref{GNAR.eq} is given by $G_j(X)=X^\top$.

An NAR$(p)$ model can be written as a stacked NAR$(1)$ process in the following way. Let  $\mathds{X}_t=(X_t^\top,X_{t-1}^\top,\dots,X_{t-p}^\top)^\top$. Then, the stacked NAR$(1)$ process corresponding to \eqref{GNAR.eq} is given by 
$
\mathds{X}_t= (\tilde A \Mm \widetilde{G(Ad}_{t-1})) \mathds{X}_{t-1}+ (e_1 \otimes I_d) \eps_t, \text{ where }
$
$$
\tilde A =
\begin{pmatrix}
A_1  & A_2  & \dots & A_{p-1} & A_p  \\
I_d & 0 & \dots & 0 & 0 \\
0 & I_d & & 0 & 0 \\
\vdots & & \ddots & \vdots & \vdots \\
0 & 0& \dots & I_d & 0
\end{pmatrix}
$$
and 
$$
\widetilde{G(Ad}_{t-1})=
\begin{pmatrix}
G_1(Ad_{t-1}) & G_2(Ad_{t-2}) & \dots &  G_{p-1}(Ad_{t-p+1}) & G_{p}(Ad_{t-p}) \\
I_d & 0 & \dots & 0 & 0 \\
0 & I_d & & 0 & 0 \\
\vdots & & \ddots & \vdots & \vdots \\
0 & 0& \dots & I_d & 0
\end{pmatrix}$$
are the corresponding matrices of the stacked NAR$(1)$ process. The process $\{\widetilde{G(Ad}_{t}), t \in \Z\}$ is denoted by $\Ad_G$.

It is also possible to handle more than one attribute at a time by simply enlarging $\Ad$. For two different attributes we can replace $Ad_t$ by the following  matrix  $\begin{pmatrix}
Ad_t & B_t & \\
C_t & Ad_t
\end{pmatrix}$, where $B_t$ and $C_t$ describe the (time-dependent) relationship between the different attributes. E.g, if there shall only be an influence between the different attributes of the same actor, we set $B_t=C_t=I_d$, or if the different attributes shall influence each other in the same way as they are influenced by their own kind, we set $B_t=C_t=Ad_t$.

Model (\ref{GNAR.eq}) is inspired by \citeasnoun{knight2016modelling}. Let $s_j$ be the order number for lag $j$, which denotes the maximal stage neighbors included for lag $j$. Let $Ad$ be a static adjacency matrix and denote by $\mathcal{N}^{(k)}(r)=\{j = 1,\dots,p: e_r^\top \mathcal{N}_k(Ad) e_j=1\}$ the $k$-stage neighbors of vertex $r$. Then, for component $i=1,\dots,d$, their autoregressive model is given in the following way 
\begin{align} \label{eq.knight}
X_{t;i}=\sum_{j=1}^p \alpha_j X_{t-j;i}+\sum_{r=1}^{s_j} \sum_{q \in \mathcal{N}^{(r)}(i)}  \beta_{j,r,q} X_{t-j;q} + \eps_{t;i}.
\end{align}
Note that $k$-stage neighborhood sets are disjoint for different $k$. Thus, the above model fits into the framework \eqref{def.NLP} in the following way
\begin{align}
\label{eq.knight.myway}
X_t= \sum_{j=1}^p (A_j \Mm (I_d + \sum_{r=1}^{s_j} \mathcal{N}_k(Ad))) X_{t-j} + \eps_t.
\end{align}

\citet{knight2019generalised} extended the model \eqref{eq.knight} to dynamic networks with potential covariates and edge weights. Apart from the covariates, their extended model fits also in the framework \eqref{GNAR.eq}. As mentioned in the case of a static network, model \eqref{GNAR.eq} can be considered as a vector autoregressive model with parameter constraints. For this case, \citet{knight2019generalised} give  conditions for stationarity and showed consistency of the least square approach. However, they give no theoretical results for the case of a dynamic network. In this work, the focus is on the dynamic case. Note that for this case model \eqref{GNAR.eq} cannot be considered as a VAR model with constraints anymore.

The condition that $\{\eps_s, s>t\}$ and $\{Ad_s, s\leq t\}$ are independent for all $t$ ensures that  \eqref{def.NLP} is a meaningful and causal representation. This condition allows that there could be an interaction between network $\Ad$ and $\Xb$ in a way  that the network at time $t$ can be influenced by $\{X_s, s\leq t\}$. That means an underlying network given by the adjacency process $\Ad$ has to fulfill only this condition, strictly stationarity and later on some dependence measure conditions on the dynamic behavior, but  we assume nothing about the inner structure of the network. Thus, it does not matter if its a sparse or dense network, or if it has properties like the small-world-phenomenon. Hence, it gives the flexibility that the time series and the underlying dynamic network can be modeled separately. One is not fixed to a specific network model as it would be the case for a joint modeling approach. Instead, the idea is that the approach described here is used to model the time series $\Xb$, and one of the several models for dynamic networks can be used to model the network $\Ad$.

If $\Ad$ is a deterministic sequence, the GNARMA model is closely related to time-varying ARMA models, which, for instance, are used in the locally stationary framework; see \cite{dahlhaus1999nonlinear, wiesel2013time}. Furthermore, if $\Ad$ is i.i.d., this framework reduces to the framework of random coefficient models, see for instance \cite{nicholls} and for the multivariate setting \cite{NICHOLLS1981185}. 
However, assuming independence between different time-points for the process $\Ad$ seems to be inappropriate in the framework of dynamic networks. Some form of influence of the recent history seems to be more reasonable, see among others \cite{blei2007statistical}.

Assumption~\ref{ass.stat} gives conditions which implies that \eqref{GNAR.eq} possesses a causal, stationary solution, see the following Lemma~\ref{lem.solution} for details. Assumption~\ref{ass.stat}a) imposes only conditions on $X$, but no restrictions on the underlying dynamic network are required. If more about the underlying network is known, e.g. its weights and sparsity setting, one may work with Assumption~\ref{ass.stat}b) which is more general but harder to verify without knowledge about the network. In Section~\ref{sec.LNAR}, a simplified model is considered in which Assumption b) can be verified under simple conditions.

\begin{asp}
\label{ass.stat}
In \eqref{GNAR.eq} let for $j=1,\dots,p,$ $\|G_j(\cdot)\|_{\max}\leq 1$ and let further one of the following hold
\begin{enumerate}[a)]
\item $\det(I-\sum_{j=1}^p |A_j| z^j)\not = 0$ for all $|z|\leq 1$,
\item $\rho(\tilde A \Mm \tilde G(\cdot))<1$, where $\tilde A, \tilde G$ denote the corresponding quantities of the stacked process.
\end{enumerate}
\end{asp}

\begin{lem} \label{lem.solution}
Under Assumption~\ref{ass.stat}, the process \eqref{GNAR.eq} possesses a stationary solution.  The solution takes the form
\begin{align}\label{eq.NAR.MAinfty}
X_t=\sum_{j=0}^\infty ( e_1 \otimes I_d)^\top \prod_{s=1}^j (\tilde A \Mm G(\widetilde{Ad}_{t-s})) ( e_1 \otimes I_d) \eps_{t-j}=:\sum_{j=0}^\infty B_{t,j} \eps_{t-j},    
\end{align}
where for all $t$ $\|B_{t,j}\|_2\leq \|\ |\tilde A|^j\|_2$.  The process has the following autocovariance function
\begin{align*}
\Gamma(h)=&\sum_{j_2=0}^\infty \sum_{j_1=0}^\infty \cov( ( e_1 \otimes I_d)^\top(\prod_{s_1=1}^{j_1} \tilde A \Mm \widetilde{G(Ad)}_{h-s_1})( e_1 \otimes I_d) \eps_{h-j_1}, \\
&( e_1 \otimes I_d)^\top \prod_{s_2=1}^{j_2} (\tilde A \Mm \widetilde{G(Ad)}_{-s_2}) ( e_1 \otimes I_d) \eps_{-j_1}), h \geq 0,
\end{align*}
and $\Gamma(h)=\Gamma(-h)^\top, h <0$. $E X_t=\sum_{j=0}^\infty E B_{0,j} \eps_{-j}=:\mu_X$.
\end{lem}
The autocovariance as well as the mean of $X$ is affected by the dynamic behavior of the underlying network. In order to get a better understanding on how the dynamic dependency of the network affects the time series, the following Lemma~\ref{lem1.DSLP} presents the autocovariance structure in the more simple case that $\{Ad_t\}$ and $\{\eps_t\}$ are mutually independent. Note that Assumption~\ref{ass.stat} implies for model \eqref{GNAR.eq} the conditions~\eqref{ass.5a},\eqref{ass.5b} of the following Lemma.

\begin{lem}\label{lem1.DSLP} Let $\{X_t\}$ be a generalized network linear process as defined in $(\ref{def.NLP})$. If 
\begin{enumerate}[i)]
\item $\{Ad_t\}$ and $\{\eps_t\}$ are mutually independent,
\item $\sum_{s=0}^\infty \left( E | B_{j,s+l} \Sigma_\eps B_{0,s}^\top | \right)+ \sum_{s_1=0}^\infty \sum_{s_2=0}^\infty |\cov \left(B_{j,s_1} \mu, B_{l,s_2} \mu \right)| <\infty$ for all $j,l \in \N$ (component-wise)   \label{ass.5a}
\item $\sum_{s=0}^\infty \left( E | B_{0,s}|\right) < \infty$ (component-wise), \label {ass.5b}
\end{enumerate}
hold, then $X_{t} = \lim_{q \to \infty }\sum_{j=0}^q B_{t,j} \eps_{t-j}$ converges component-wise in the $L_2$-limit, and the autocovariance function  is given by $\Gamma_X(h)=\Gamma_X(-h)^\top$, and 
\begin{align}
\Gamma_X(h)=\sum_{s=0}^\infty E \left(B_{h,s+h} \Sigma_\eps B_{0,s}^\top\right) +\sum_{j=0}^\infty \sum_{s=0}^\infty \cov\left( B_{h,j} \mu, B_{0,s} \mu\right), h\geq 0, \label{acf.dslp}
\end{align}
and the mean function by $\mu_X = \sum_{j=0}^\infty E B_{0,j} \mu$.
\end{lem}
The latter term of the autocovariance function, $\sum_{j=0}^\infty \sum_{s=0}^\infty \cov\left( B_{h,j} \mu, B_{0,s} \mu\right)$, comes only into play for non-centered innovations and is driven by the linear dependency structure of the network. Consequently, it can be seen that the linear dependency of the network directly influences the linear dependency of the process $\Xb$. As a consequence, even a network moving average process of order $q$ may possess a nonzero autocovariance for lags higher than $q$. In order to better understand this, consider a small toy example with three vertices and two possible edges, $(1,3)$ and $(2,3)$, and only one edge is present at a time. Let $\{ e_t, t \in \Z\}$ be i.i.d. random variables with uniform distribution on $[0,1]$, i.e., $e_1 \sim \mathcal{U}[0,1]$. Which edge is present at time $t$ is given by the random variables $(e_t)$ in the following way. If $Ad_{t-1;13}=1$, then if $e_t>0.05$, then $Ad_{t;13}=1$ else $Ad_{t;23}=1$. If $Ad_{t-1;13}=0$ (that means $Ad_{t-1;23}=1$), then if $e_t>0.95$, then $Ad_{t;13}=1$ else $Ad_{t;23}=1$. Consequently, we flip in this network between the edges $(1,3)$ and  $(2,3)$, and if one edge is present at time $t$, it is more likely (with probability $0.95$) that it is present at time $t+1$ than flipping to the other edge. We have dependency between different time points as well as between edges. $\eps_1 \sim \mathcal{N}( \mu , I_3),$ and $ \mu=(10,-10,0)^\top$. Let $\Xb$ be given by 
\begin{align}
X_{t}= Ad_{t-1}^\top \eps_{t-1} + \eps_{t}= Ad_{t-1}^\top X_{t-1} + \eps_{t}, \text{ where } Ad_\cdot^\top=
\begin{pmatrix}
0 & 0 & 0 \\
0 & 0 & 0 \\
* & * & 0 
\end{pmatrix}. \label{example.flip}
\end{align} Thus, $\Xb$ is a network moving average process of order $1$, and the influence goes in the direction of the edges. Since no edge goes into vertex $1$ or $2$, $\{X_{t;1}, t \in \Z\}$ and $\{X_{t;2}, t \in \Z\}$ are white noise. This can be also seen in the autocovariance function, which is displayed in its two parts in Figure \ref{ex.flip.acf}. The left-hand side figures display the first part; $\sum_{s=0}^\infty E \left(B_{h,s+h} \Sigma_\eps B_{0,s}^\top\right)$. The dependency of the network has no influence on the first part, thus, this part would remain the same if a static model was considered where $\Ad$ is replaced by its expected value. That is why  this part of the autocovariance function has the structure one expects from a vector moving average (VMA) process of order $1$. The right-hand side figures display the second part of the autocovariance function in Lemma~\ref{lem1.DSLP}, $\sum_{j=0}^\infty \sum_{s=0}^\infty \cov\left( B_{h,j} \mu, B_{0,s} \mu\right)$.  As already mentioned, this part is completely driven by the linear dependence structure of the network. For the two edges, we have the following linear dependency: $\cov(Ad_{t+h;23},Ad_{t;23})=\cov(Ad_{t+h;13},Ad_{t;13})=0.9^h/4, \cov(Ad_{t+h;23},Ad_{t;13})=\cov(Ad_{t+h;13},Ad_{t;23})=-0.9^h/4$.
This explains the geometric decay in the autocovariance function of the third component of $\Xb$, whereas the magnitude of the autocovariance function of the third component is mainly given by the difference of the  mean of the innovations of the first two components. Hence, a greater difference of the innovations mean makes it harder to identify the linear dependency, which means the first part of the autocovariance function in Lemma~\ref{lem1.DSLP}, between components $1$ and $3$, or $2$ and $3$ respectively. In this particular example with mean $ \mu=(10,-10,0)^\top$, no linear dependency between the different components can be identified for moderate sample sizes. A sample autocorrelation function as well as a realization of the third component of $\Xb$ is displayed in Figure \ref{ex.flip.real} for a sample size $n=500$. Instead, looking from the perspective of the classical time series analysis, the sample autocorrelation function looks like we have three uncorrelated components where the first two components are white noises and the third could be an AR$(1)$ process. Hence, this examples gives two important aspects to keep in mind: firstly, the linear dependency of the network can influence the linear dependency of the time series directly. Secondly, the problem that the autocovariance function may not suffice to identify network linear processes such as network autoregressive models should be kept in mind.

\begin{figure}[H]
\includegraphics[width=0.495\textwidth]{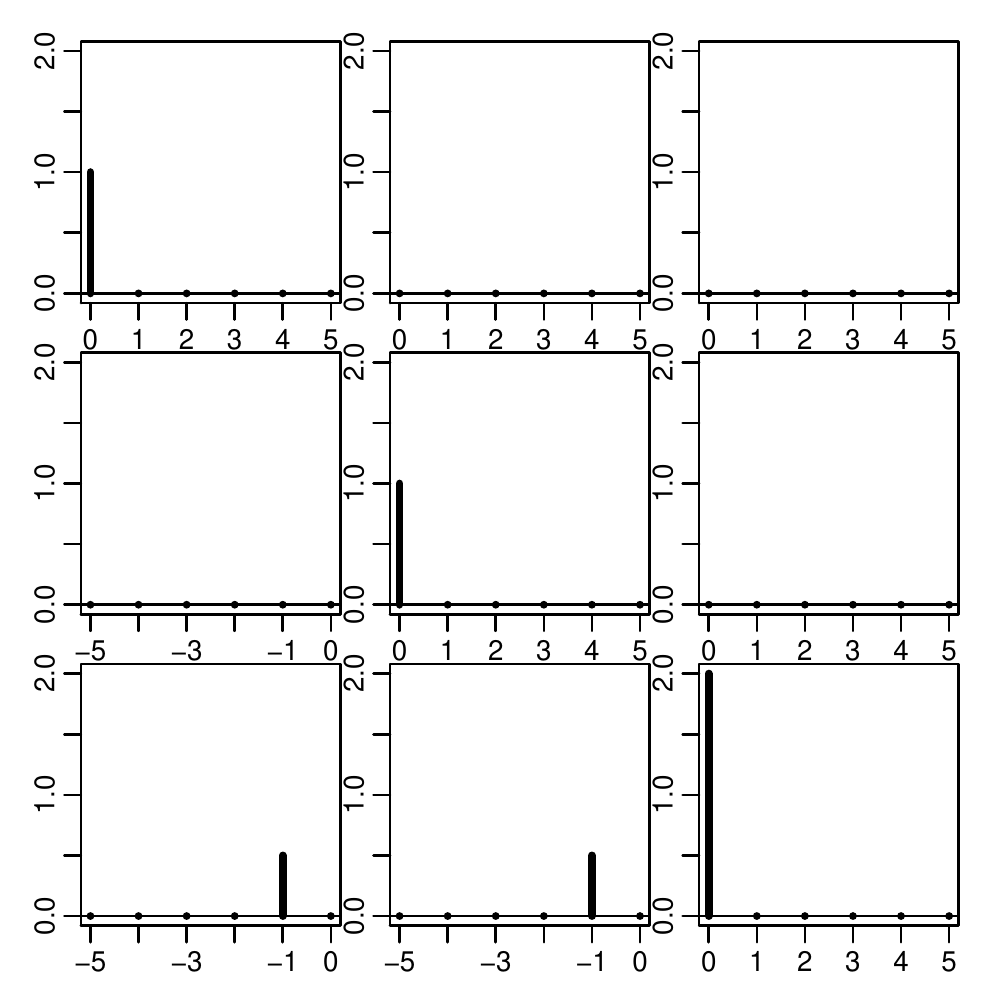}
\hfill
\includegraphics[width=0.495\textwidth]{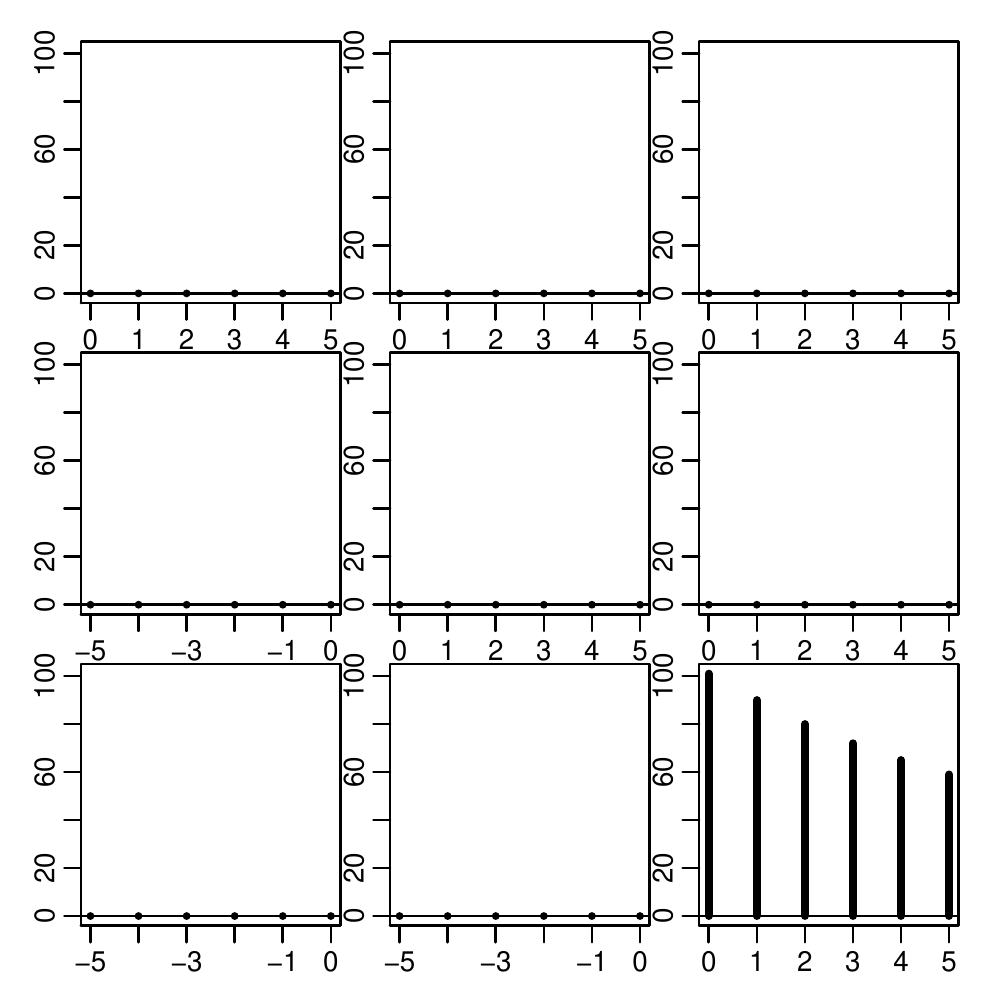}
\caption[Autocovariance function $(\ref{acf.dslp})$ of process $(\ref{example.flip})$]{Left-hand-side ($\sum_{s=0}^\infty E(B_{h,s+h} \Sigma_\eps B_{0,s}^\top)$; left figure) and right-hand-side ($\sum_{j=0}^\infty \sum_{s=0}^\infty \cov( B_{h,j} \mu, B_{0,s} \mu)$; right figure) of the autocovariance function $(\ref{acf.dslp})$ of process $(\ref{example.flip})$} \label{ex.flip.acf}
\end{figure}

\begin{figure}[H]
\includegraphics[width=0.495\textwidth]{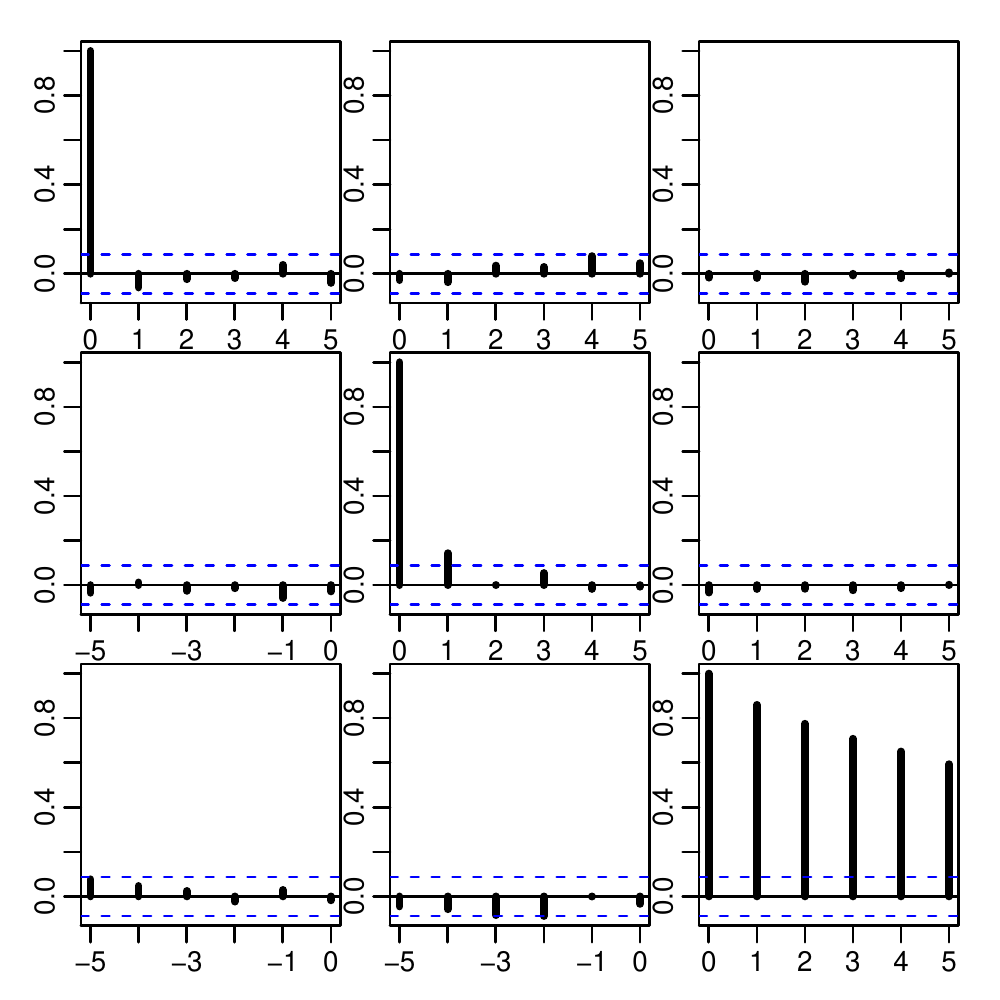}
\hfill
\includegraphics[width=0.495\textwidth]{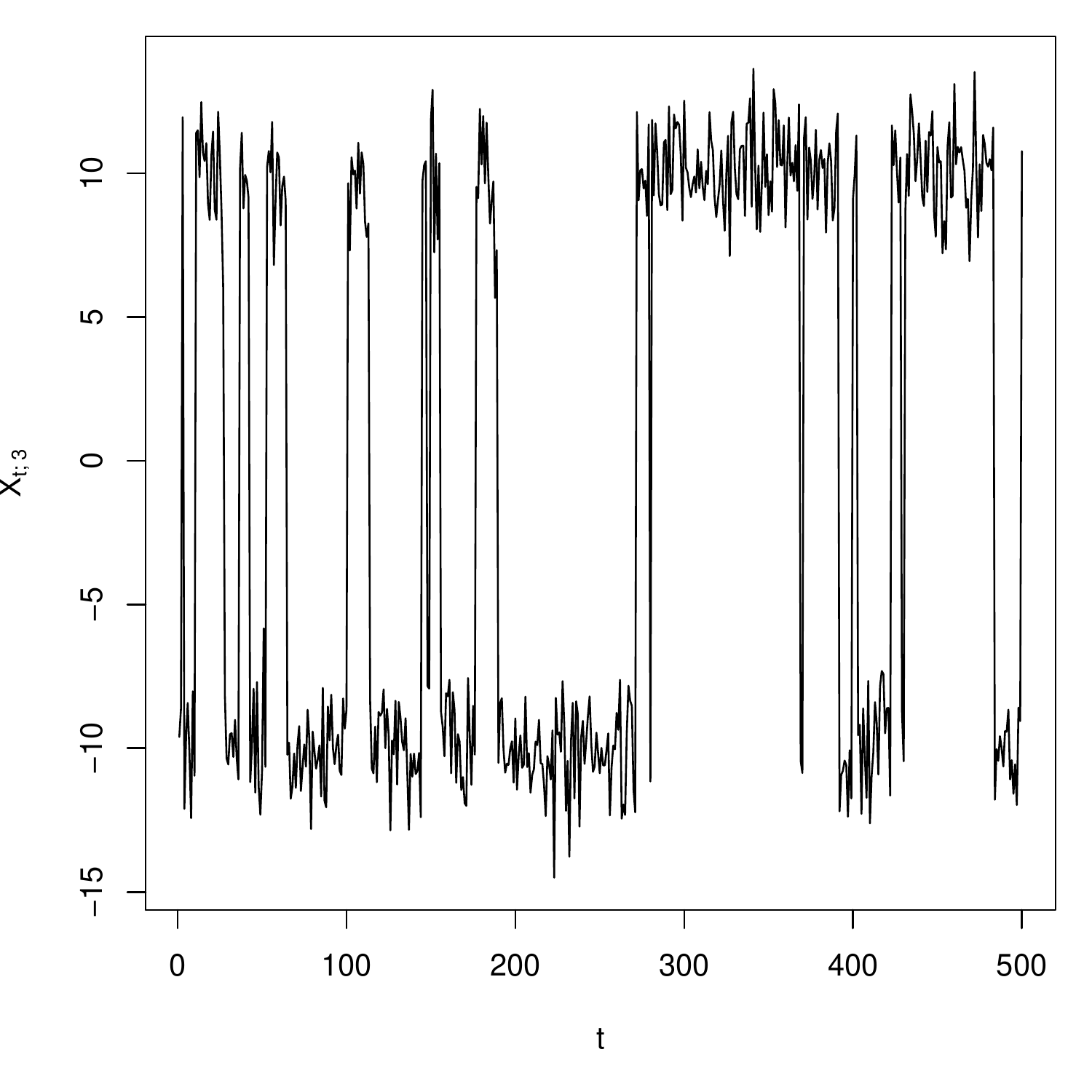}
\caption[Sample ACF and realization of process $(\ref{example.flip})$]{Sample autocorrelation function and realization of the third component of process $(\ref{example.flip})$, based on $n=500$} \label{ex.flip.real}
\end{figure}

\section{Statistical results}
\label{stat.result}
In this section we focus on the estimation of model \eqref{GNAR.eq}. As seen in the example in Section \ref{sec.3.2}, the autocovariance function is not helpful to identify such models. Note further that even if $\Ad$ is Markovian, an NAR$(1)$ process can generally not be written as a Hidden Markov model (HMM). This is because $\Xb$ given $\Ad$ is not a sequence of conditionally independent variables and cannot be written as a noisy functional of $Ad_{t-1}$ only, which is required by a HMM, see among others \cite{bickel1998asymptotic} for details on HMM. Consequently, techniques used for HMM cannot be applied here.
Instead, the same setting as in \citet{knight2016modelling,knight2019generalised,zhu2017network} is considered. Thus, the process $\Xb$ as well as the network $\Ad$ is observed leading to observations $X_{1},\dots,X_{n}$ and $Ad_1,\dots,Ad_{n-1}$. In such a setting, the consistency of a least square estimate as well as asymptotic normality for model \eqref{GNAR.eq} is shown in the first subsection. The results are presented under general assumptions and the asymptotic setting that $d$ is fixed and $n\to \infty$. Later on, we give dependence measure conditions for the underlying dynamic network such that the general assumptions hold. In the second subsection, a simplified version of model \eqref{GNAR.eq} is considered. This simplified model is suited for high-dimensional cases, and consequently, the theoretical estimation results are presented nonasymptotically.

\subsection{Network autoregressive models} \label{sec.NAR.stat}
Networks usually come together with some form of sparsity, see among others Section 3.5 in \cite{Kolaczyk:2009}. This means that a vertex has only a connection to a limited number of other vertices and $E \sign(Ad_1)$ could have some zero entries or  might even be a sparse matrix. Thus, $E \sign(G_j(Ad_1)), j=1,\dots,p$ might be sparse matrices as well. That means the number of parameters of model \eqref{GNAR.eq} is given by $\sum_{j=1}^p\|\veco(E |\sign(G_j(Ad_1))|)\|_0\leq pd^2$ and depends on the sparsity of the underlying network. Let $\mathcal{I}{(r)}=\{\tilde i=i+(j-1)d, i=1,\dots,d,j=1,\dots,p :  1/n\sum_{t=p+1}^n e_r^\top| G_j(Ad_{t-j})|e_i>0\}, r=1,\dots,d $ be a set of indices corresponding to the non-zero coefficients of $\sum_{t=p+1}^n|(G_1(Ad_{t-1}):\dots:G_p(Ad_{t-p}))|$ and $\mathcal{I}{(r)}_E=\{\tilde i=i+(j-1)d, i=1,\dots,d,j=1,\dots,p :   E e_r^\top| G_j(Ad_{1})|e_i>0\}, r=1,\dots,d$ is the corresponding population quantity. Note that $\mathcal{I}{(r)} \subseteq  \mathcal{I}{(r)}_E$ for all $t$. Only parameters corresponding to indices of the set $\mathcal{I}{(r)}_E$ are well defined in the sense that they have an influence on the process. We set the other parameters, meaning those corresponding to indices of the set $(\mathcal{I}{(r)}_E)^C$, to zero. Recall that  $I_{d;-I}\in \R^{(d-|I|)\times d}$ denotes a $d$-dimensional identity matrix without rows $i \in I$ and $I_{d;I}=I_{d;-I^C}.$ Let for $r=1,\dots,d,$ $w_r=I_{dp;\mathcal{I}{(r)}}(e_r^\top A_1, \dots,e_r^\top A_p)^\top$ and $$Y_{t-1}^{(r)}=I_{dp;\mathcal{I}{(r)}}( (e_r^\top G_1(Ad_{t-1})) \Mm X_{t-1},\dots,(e_r^\top G_p(Ad_{t-p})) \Mm X_{t-p})^\top.$$ Then for $t=p+1,\dots,n$ and $r=1,\dots,d,$ we can write \eqref{GNAR.eq} as
\begin{align}\label{GNAR.comp}
X_{t;r}= w_r^\top Y_{t-1}^{(r)}+\eps_{t;r}.
\end{align}
Thus, $w_r$ and $\mu_r=E \eps_{t;r}$ can be estimated by using the following least square approach given by 
$
\argmin_{\hat w_r,\hat \mu_r} \sum_{t=p+1}^n (X_{t;r}- \hat w_r^\top Y_{t_1}-\hat \mu_r)^2.
$
For component $r$, this leads to the following linear system:
\begin{align} \label{eq.sys.est}
\begin{pmatrix}
\sum_{t=p+1}^n X_{t;r} \Yr_{t-1}- \frac{1}{n-p} \sum_{t_1,t_2=p+1}^n \Yr_{t_1-1} X_{t_2;r} \\
\sum_{t=p+1}^n X_{t;r}
\end{pmatrix}
= \nonumber\\
\begin{pmatrix}
0 &  \sum_{t=p}^{n-1} \Yr_t (\Yr_t)^\top - \frac{1}{n-p} \sum_{t_1,t_2=p}^{n-1} (\Yr_{t_1})  (\Yr_{t_2})^\top \\
n-p &  \sum_{t=p}^{n-1} (\Yr_t)^\top
\end{pmatrix}
\begin{pmatrix}
\hat \mu_{r} \\
\hat w_r
\end{pmatrix} 
.
\end{align}

We show the consistency of the least square estimators under general assumptions. Later on, we specify a dependence concept for the underlying network which ensures these general assumptions, see Lemma~\ref{lem.physical}.

\begin{asp}\label{ass.conv} For all $r=1,\dots,d,$ we have, as $n\to \infty$, 
\begin{enumerate}
\item $E Y_1^{(r)}=\mu_{\Yr}$,
$
1/n \sum_{s=1}^n \Yr_s =\mu_{\Yr}+O_P(1/\sqrt n),
$
and $\Gamma_{\Yr}(0)=\var(Y_1^{r})$, $\| \Gamma_{\Yr}(0)\|_2$, $\| \Gamma_{\Yr}(0)^{-1}\|_2<\infty,
$
$$
\frac{1}{n} \sum_{t=1}^n (\Yr_t-\mu_{\Yr})(\Yr_t-\mu_{\Yr})^\top
= \var(\Yr_1)+O_P(1/\sqrt n).$$
\item 
$$
\frac{1}{n} \sum_{t=1}^n \eps_{t;r}= \mu_r+O_P(1/\sqrt n),
$$
and
$$
\frac{1}{n} \sum_{t=1}^n  (\Yr_t-\mu_{\Yr})(\eps_{t+1;r}-\mu_r) = \cov(\Yr_0,\eps_1)+O_P(1/\sqrt n)=O_P(1/\sqrt n).
$$
\item For all $s=1,\dots,d, E|\Yr_{0;s}|^4<\infty$ and $E| \eps_{0;s}|^4<\infty$.
\end{enumerate}
\end{asp}

\begin{asp} \label{ass.norm} For all $r=1,\dots,d,$ we have, as $n\to \infty$,
$$
\frac{1}{\sqrt{n}} \sum_{t=1}^n (\eps_{t;r}-\mu_r) (\Yr_{t-1}-\mu_{\Yr})^\top \overset{D}{\to} \mathcal{N}(0,\Sigma_\eps),
$$
where $\Sigma_\eps=\var(\eps_{0;r}) \Gamma_{\Yr}(0)=\Sigma_{\eps;rr}\Gamma_{\Yr}(0)$.
\end{asp}

\begin{thm}\label{thm.cons.NAR}
Under Assumption~\ref{ass.conv} we have for $r=1,\dots,d,$  $\hat \mu_r=\mu_r + O_P(1/\sqrt{n}), \hat w_r=w_r+O_P(1/\sqrt{n})$. If additionally Assumption~\ref{ass.norm} holds, we have, as $n \to \infty$,
$$
\sqrt{n} \begin{pmatrix}
\hat \mu_r - \mu_r\\
\hat w_r - w_r
\end{pmatrix} \overset{D}{\to} \mathcal{N}\left(0, \Sigma_{\eps;rr} \begin{pmatrix}
1+\mu_{\Yr}^\top \Gamma_{\Yr}(0)^{-1} \mu_{\Yr}  & 0 \\
0 & \Gamma_{\Yr}(0)^{-1}
\end{pmatrix}\right).
$$
Furthermore, we have for $r,s=1,\dots,p,$ as $n \to \infty$,
$$n\cov(\mu_r,\mu_s)\to \Sigma_{\eps;rs}(1+ \mu_{\Yr}^\top \Gamma_{\Yr}(0)^{-1} \cov(\Yr_1,Y_1^{(s)}) \Gamma_{Y^{(s)}}(0)^{-1} \mu_{Y^{(s)}}),$$ 
$$n \cov(w_r,w_s)\to \Sigma_{\eps;rs} \Gamma_{\Yr}(0)^{-1} \cov(\Yr_1,Y_1^{(s)}) \Gamma_{Y^{(s)}}(0)^{-1}$$ and $n \cov(w_r, \mu_s)\to 0.$
\end{thm}

The results of Theorem \ref{thm.cons.NAR} can be used to forecast the process $\Xb$. If $Ad_n$ is observed, then let $Y_{n}^{(r)}=I_{dp;\mathcal{I}{(r)}}( (e_r^\top G_1(Ad_{n})) \Mm X_{n},\dots,(e_r^\top G_p(Ad_{n-p+1})) \Mm X_{n-p+1})^\top$ and a one-step ahead forecast of 
$X_{n+1}$ is given by 
${\hat X}^{(1)}_{n+1;r}=\hat w_r Y_n^{(r)}+\hat \mu_r,r=1,\dots,d.$ Since $\{\eps_t\}$ is i.i.d. and $\hat w_r, \hat \mu_r$ are $\sqrt n$ consistent, we have $E(X_{n+1;r}-\hat X^{(1)}_{n+1:r})^2=e_r^\top\Sigma_\eps e_r+O(1/\sqrt{n})$. If $Ad_n$ is not observed, $Ad_n$ itself needs to be predicted first. This could be done by fitting a dynamic network model to $Ad_1,\dots,Ad_{n-1}$ and using this model to predict $Ad_n$. An $h$-step ahead forecast can be done recursively, which means performing a one-step ahead forecast based on the observations and the results of the $h-1,\dots,1$-step ahead forecasts. 

Assumption~\ref{ass.conv} mainly requires a $\sqrt{n}$ conversion rate of the first and second sample moments of $\Xb$. An absolutely summable autocovariance function of $\Xb$ is sufficient for the convergence of the first sample moments. As pointed out in Lemma~\ref{lem.solution} and Lemma~\ref{lem1.DSLP}, the autocovariance of $\Xb$ depends on the dependency structure of $\Ad$. For simplicity, consider the following network moving average process $X_t= Ad_{t-1} \eps_{t-1} +\eps_t, E \eps_1=\mu, \var \eps_1=\Sigma_\eps$  where $\Ad$ and $\{\eps_t\}$ are independent. Following Lemma~\ref{lem1.DSLP}, we obtain
$\sum_{h=0}^\infty \Gamma_X(h)=(\Sigma_\eps+E(Ad_1 \Sigma_\eps Ad_1^\top)+(E Ad_1 \Sigma_\eps)+\sum_{h=0}^\infty\cov(Ad_{h}\mu,Ad_0\mu)$. Hence, even for this simple moving average process, a summmable autocovariance function can be obtained only if $\Ad$ possesses some sort of short-range dependence. Several general dependency concepts exists which could describe a short-range dependency structure such as mixing\cite{bradley2007}, some weak dependency concepts\cite{doukhan1999new} or physical dependence\cite{wu2005nonlinear,wu2011asymptotic,wu2011gaussian,liu2009strong}. Since the concept of physical dependence works well also in the high-dimensional case, see among others \cite{zhang2017gaussian,zhang2018gaussian}, this concept is used to quantify the dependence structure of $\Ad$. 

To elaborate, let $\{\xi_t\}$ be a sequence of i.i.d. random vectors of dimension $\tilde d$ such that $\{\Xi_t=(\eps_t,\xi_t)\}$ is also an i.i.d. sequence. Furthermore, let $Ad_t=H(\Xi_t,\Xi_{t-1},\dots)$, where $H$ is some measurable function to $[-1,1]^{d\times d}$.  Denote by $\Xi_t^\prime$ an i.i.d. copy of $\Xi_t$ and let for some $q>0$
$\delta_q(\Ad,j)=\max_{r,i=1,\dots,d} \linebreak\|e_r^\top(Ad_j-Ad_j^*) e_i\|_{E,q},$ where $Ad_j^*=H(\Xi_j,\Xi_{j-1},\dots,\Xi_{1},\Xi_0^\prime,\Xi_{-1},\Xi_{-2},\dots)$ is a coupled version of $Ad_j$ with $\Xi_0$ in the latter being replaced by $\Xi_0^\prime$. Since $H$ is a function to $[-1,1]^{d\times d}$, $\delta_q(\Ad,j)<\infty $  for $q\geq 1$. Furthermore, let $\Delta_q(\Ad)=\sum_{j=0}^\infty \delta_q(\Ad,j).$ The process $\Ad$ is denoted as $q$-stable if $\Delta_q(\Ad)<\infty$. This property still holds for some nonlinear transformations, see \cite{wu2005nonlinear,wu2006unit}. E.g. consider a polynomial transformation given by  $\{g(Ad_t)=Ad_t^k\}$. Note that for some matrices $A,B$ we have $A^k-B^k=\sum_{s=0}^{k-1} A^s(A-B) B^{j-1-s}$. If $\|\Ad\|_\infty \leq C$, then $\delta_{q/2}(\{g(Ad_t)\},j)=\delta_q(\Ad,j) C^{k-1} k$. Without assuming any sparsity, an upper bound is given by $C\leq d$. This dependency concept covers a wide range of processes among them many nonlinear time series, see \cite{wu2005nonlinear,wu2011asymptotic,wu2011gaussian} for examples. Furthermore, this concept includes nonlinear Markov chains, meaning $\Ad$ can be given by $Ad_t=H(Ad_{t-1},\Xi_t)$.  \citet{zhang2018gaussian} pointed out that a stable process is obtained if $H$ possesses some form of Lipschitz-continuity. Then, $\delta_q(\Ad,j)=O(\rho^j)$ for some $\rho \in (0,1)$, see Example 2.4 in \cite{zhang2018gaussian} or Example 2.1 in \cite{chen2013covariance} for details. A stable vector autoregressive process possesses also such geometrically decaying physical dependence coefficients, see among others Example 2.2 in \cite{chen2013covariance}. Note that many dynamic network models, e.g. Temporal ERGMs \cite{hanneke2010discrete}, are Markovian.

\begin{lem}\label{lem.physical}
If Assumption~\ref{ass.stat}a) holds, 
$\Ad_G=\{\widetilde{G(Ad}_t)\}$ is $2q$-stable, and $\max_r \|\eps_{0;r}\|_{E,2q}<\infty$ for some $q\geq 1$, then $\Xb$ is $q$-stable. If the above conditions hold for $q\geq 4$, then Assumption~\ref{ass.conv} and \ref{ass.norm} hold.
\end{lem}


\subsection{Network autoregressive models for large dimension} \label{sec.LNAR}
The number of parameters in model \eqref{GNAR.eq} is of size $O(pd^2)$. If the underlying network is not very sparse, a reasonable estimate could be only obtained if $d\ll n$. Hence, in order to handle high-dimensional cases, meaning $d$ is of the same order as $n$ or even larger, we follow \citet{knight2019generalised} and simplify model \eqref{GNAR.eq}. For each component, the influence of the own lagged components is modeled separately, thus, we set  $e_s^\top G_j(\cdot) e_s=0$ for all $s=1,\dots,d,j=1,\dots,p$.  Then, the simplified model is given by
\begin{align} \label{eq.LNAR}
X_{t;r}=\sum_{j=1}^p \alpha_{j,r} X_{t-j;r}+ \beta_{j,r} e_r^\top G_j(Ad_{t-j}) X_{t-j} + \eps_{t;r}, r=1,\dots,p,
\end{align}
where $\alpha_{j,r},\beta_{j,r}\in \R, j=1,\dots,p,r=1,\dots,d$ and $E\eps_{t;r}=\mu_r$. 
Hence, this simplified model possesses in total only $d(2p+1)$ parameters or more precisely only $2p+1$ parameters for each component of the time series independently of the dimension. The parameter $\alpha$ quantifies the linear influence of the same component and $\beta$ the linear influence of the other components.  
Note that model \eqref{eq.LNAR} can be written as
\begin{align}
X_t=\sum_{j=1}^p
\begin{pmatrix}
\alpha_{j,1} & \beta_{j,1} & \dots & \beta_{j,1} & \beta_{j,1} \\
\beta_{j,2} & \alpha_{j,2} & \beta_{j,2} & \dots & \beta_{j,2} \\
\vdots  & \ddots & \ddots &  \ddots & \vdots \\
\beta_{j,d-1} & \dots & \beta_{j,d-1} & \alpha_{j,d-1} & \beta_{j,d-1} \\
\beta_{j,d} & \dots & \beta_{j,d} & \beta_{j,d} & \alpha_{j,d}
\end{pmatrix}\Mm(I_p + G_j(Ad_{t-j}))
X_{t-j} + \eps_t, \label{eq.LNAR.NAR}
\end{align}
and consequently fits into the framework \eqref{GNAR.eq}. We denote the process as Large Network AutoRegression (LNAR) and the coefficient matrices occurring in \eqref{eq.LNAR.NAR} by $A_{j,\alpha,\beta}, j=1,\dots,p$. Since a LNAR is an NAR process, a stationary solution is given by Lemma~\ref{lem.solution} if $\det(I-\sum_{j=1}^p |A_{j,\alpha,\beta}| z^j)\not = 0$ for all $|z|\leq 1$ or $\rho(\tilde A \Mm \tilde G(\cdot))<1$. If no restrictions on the underlying network are imposed, then the first condition implies that in order to obtain a stationary solution,  the parameter space depends on $d$.  This is not the case we would like to consider here, which is why conditions on the underlying network are imposed. We require that $\|G_j(\cdot)\|_\infty \leq 1, j=1,\dots,p$, which means that the sum of  weights of the edges going into a vertex  does not grow with the dimension $d$. To simplify notation, we bound the sum of weights by $1$.  \citet{knight2019generalised} require a similar condition in the case of a static network. Under this condition, we obtain a stationary solution if $\max_{r=1,\dots,p} \sum_{j=1}^p |\alpha_{j,r}|+|\beta_{j,r}|\leq C_\lambda<1$, see the following Lemma~\ref{lem.stat.LNAR}. Note that under the same conditions \citet{knight2019generalised} obtain a stationary solution in the case of a static network. 

\begin{lem}
\label{lem.stat.LNAR}
If $\|G_j(\cdot)\|_\infty \leq 1$ for $j=1,\dots,p,$ and $\max_{r=1,\dots,p} \sum_{j=1}^p |\alpha_{j,r}|+|\beta_{j,r}|\leq C_\lambda<1$, then \eqref{eq.LNAR} fulfills Assumption~\ref{ass.stat}b) and 
possesses a stationary solution.  The solution takes the form
\begin{align}\label{eq.LNAR.MAinfty}
X_t=\sum_{j=0}^\infty ( e_1 \otimes I_d)^\top \prod_{s=1}^j (\tilde A_{\alpha,\beta} \Mm G(\widetilde{Ad}_{t-s})) ( e_1 \otimes I_d) \eps_{t-j}=:\sum_{j=0}^\infty B_{t,j} \eps_{t-j},  
\end{align}
where 
\begin{align*}
    &\widetilde{G(Ad)}_{t-1})=\\
&\begin{pmatrix}
I_d+G_1(Ad_{t-1}) & I_d+ G_2(Ad_{t-2}) & \dots &  I_d+ G_{p-1}(Ad_{t-p+1}) &I_d+ G_{p}(Ad_{t-p}) \\
I_d & 0 & \dots & 0 & 0 \\
0 & I_d & & 0 & 0 \\
\vdots & & \ddots & \vdots & \vdots \\
0 & 0& \dots & I_d & 0
\end{pmatrix}.
\end{align*}

Furthermore, $\rho(|\tilde A_{\alpha,\beta} \Mm \tilde G(\cdot)|)\leq C_\lambda^{1/p}.$
\end{lem}

For component $r=1,\dots,d$, let $w_r=(\alpha_{1,r},\beta_{1,r},\dots,\alpha_{p,r},\beta_{p,r})^\top \in \R^{2p}$ and $Y_{t-1}^{(r)}=(X_{t-1;r},e_r^\top G_1(Ad_{t-1})X_{t-1},\dots,X_{t-p;r},e_r^\top G_p(Ad_{t-p})X_{t-p})^\top$. Then, \eqref{eq.LNAR} can be written as $X_{t;r}= w_t^\top Y_{t-1}^{(r)}+\eps_{t;r}$. This is the same framework as in Section~\ref{sec.NAR.stat}, and the linear system \eqref{eq.sys.est} gives a least square estimate.
To cover a high-dimensional setting, we study the theoretical properties of this estimator in a nonasymptotic framework as it is done in the high-dimensional vector autoregressive case, see among others \cite{basu2015}. We make use of the Nagaev inequality for dependent variables, see Theorem 2 in \cite{liu2013probability}, to formulate nonasymptotic error bounds. Again, the physical dependency concept is used to quantify the dependency structure of $\Ad$, see the following Assumption~\ref{ass.LNAR.qstable}. 

\begin{asp} \label{ass.LNAR.qstable}
For $j=1,\dots,p$ let $\|G_j(\cdot)\|_\infty\leq 1$  and let $\Ad_{G\ind}=\{\max_{j=1,\dots,p} \linebreak \max_r e_r^\top|G_j(Ad_t)|\ind, t \in \Z\}$ be $2q$-stable with $\sum_{k=1}^\infty k \delta_{2q}(\Ad_{G\ind},k)\leq C_G$, where $C_G  <\infty$ is some constant. 
Furthermore, let $\max_{r=1,\dots,d} \sum_{j=1}^p |\alpha_{j,r}|+|\beta_{j,r}|\leq C_\lambda^p<1$ and 
$\||\tilde A_{\alpha,\beta} \Mm \tilde G(\cdot)|^j\|_\infty \leq C_A C_\lambda^j$. The constants appearing here do not depend on the dimension $d$. 
\end{asp}
Note that $\| G(\cdot) \|_\infty\leq 1$ and $\max_{r=1,\dots,d} \sum_{j=1}^p |\alpha_{j,r}|+|\beta_{j,r}|\leq C_\lambda^p<1$ implies $\||\tilde A_{\alpha,\beta} \Mm \tilde G(\cdot)|^j\|_\infty \leq \||\tilde A_{\alpha,\beta} \Mm \tilde W|\|_\infty^j\leq 1$. Furthermore, we have a bound for the largest eigenvalue $\rho(|\tilde A_{\alpha,\beta} \Mm \tilde W|)\leq C_\lambda<1$, see the proof of Lemma~\ref{lem.stat.LNAR}. If 
$G_j(\cdot)=W_j\Mm \cdot,$ where $W_j \in [-1,1]^{d \times d}$ and $\|W_j\|_\infty\leq 1$ for all $j$, then $\delta_{2q}(\Ad_{G\ind},k)=\|\max_r\max_j e_r^\top (G_j(Ad_k)-G_j(Ad_k^*))\ind\|_{E,2q}= \| \max_r \max_j e_r^\top(W_j \Mm (Ad_k-Ad_k^*))\ind\|_{E,2q}\leq \max_j \|W_j\|_\infty \delta(\Ad,k)\leq \delta(\Ad,k)$.
If $\Ad_{G\ind}$ possesses geometrically decaying physical dependence coefficients, then $\sum_{k=1}^\infty k \delta_{2q}(\Ad_{G\ind},k)\leq C_G<\infty$.

This Assumption implies that $\Xb$ as well as $\{Y_t^{(r)}, t \in \Z\}$ are $q$-stable and their physical dependency quantity $\sum_{j=0}^\infty \delta_q(\cdot,j)$ can be bounded independently from the dimension $d$, see Lemma~\ref{lem.physical.LNAR} for details.

\begin{lem}\label{lem.physical.LNAR}
If Assumption~\ref{ass.LNAR.qstable} holds, and $\max_{i} \| \eps_{0,i}\|_{E,2q}<\infty$, then $\Xb$ generated by model \eqref{eq.LNAR} is $q$-stable and 
$\sum_{j=0}^\infty \delta_q(\{\max_{r=1,\dots,d} X_{t;r}, t \in \Z\},j)\leq \max_{i} \| \eps_{0,i}\|_{E,2q} C_A/(1-C_\lambda)(C_A/C_\lambda/(1-C_\lambda) \sum_{j=0}^\infty \delta_{2q}(\Ad_{G\ind},j)+1)$. Furthermore, we have for $k=1,\dots,2p$,
$
\sum_{j=0}^\infty j \delta_q(\{\max_r  e_k^\top Y_{t}^{(r)},t \in Z\},j)\leq 
C_{\delta_Y},
$
where 
$$
C_{\delta Y}= \max_i \| \eps_{0,i}\|_{E,2q} \frac{C_A}{1-C_\lambda}\Big[\frac{C_\lambda}{1-C_\lambda}+C_G(1+\frac{C_A(2-C_\lambda)}{(1-C_\lambda)^2})\Big].
$$
\end{lem}

With this results, we can formulate the nonasymptotic error bounds. In order to handle a high-dimensional setting an import result is to obtain an error bound which grows only moderately with $d$, the dimension of the process. Note that in contrast to the estimation of a high-dimensional VAR system, e.g. \cite{kock2015oracle}, the dimension of the parameter vector does not depend on $d$. This enables us to obtain an error bound which does not depend on $d$ at all, see the following Theorem~\ref{thm.LNAR} for details.  

\begin{thm}\label{thm.LNAR}
Under Assumption~\ref{ass.LNAR.qstable} we have for component $r=1,\dots,d,$ of model \eqref{eq.LNAR}, and the estimators given by the linear system \eqref{eq.sys.est} for some $y \in \R$ with probability of at least $(1-c_q(n-p)^{1-q}y^{-q}-(c_q^\prime+2) \exp(-c_q (n-p)y^2))^4=:C_q(n,y)^4$, where $c_q,c_q^\prime$ are constants depending on $q$ only, the following error bounds
\begin{align}
\|\hat w_r-w_r\|_1\leq& y\frac{\sqrt{2p} C_{\delta Y}(\|\eps_{0;r}\|_{E,q}+C_{\delta Y}y+\mu_r +\|\mu_{\Yr}\|_1)}{\rho(\Gamma^{-1})-y 2p C_{\delta Y}(2C_A/(1-C_\lambda)+2\|\mu_{\Yr}\|_1+yC_{\delta Y}) }, \\
|\hat \mu_r - \mu_r| \leq& (\|\mu_{\Yr}\|_1+y C_{\delta Y})\|\hat w_r-w_r\|_1+y C_{\delta Y}.
\end{align}
\end{thm}

For $y=o(1/\sqrt n)$,  the error tends to zero but the probability still faces $1$ with increasing $n$. This rate is independent of the dimension $d$. This enables us to use LNAR for forecasting in a high-dimensional framework. The forecasting procedure is analogue to the one for the NAR approach, see the end of the previous subsection.

\section{Numerical Examples} \label{Simulation}
In this section, the forecasting performance of the models presented in Section~\ref{sec.3.2} is investigated in finite samples. For a low-dimensional example and a high-dimensional example, we forecast $X_{n+1},\dots,X_{n+h}$ based on observations $X_1,\dots,X_n$ and $Ad_1,\dots,Ad_{n-1}$, where $h=4$. The performance is measured by computing the mean squared error (MSE) averaged over all components via a Monte Carlo simulation using $B=1000$ repetitions, meaning $$MSE(\hat X_{n+h}^{(h)})\approx 1/d \sum_{j=1}^d 1/B \sum_{i=1}^B (X_{n+h,i;j}-\hat X_{n+h,i;j}^{(h)})^2,$$ where $X_{n,i;j}$ denotes the $j$th component of the $n$th observation of the $i$th Monte Carlo sample. In the following, we denote the approach using model \eqref{GNAR.eq} as NAR and the approach using model \eqref{eq.LNAR} as LNAR. As a benchmark, we use a vector autoregressive model given by $X_t=\sum_{j=1}^p A_j X_{t-j} +\eps_t$, where $A_1,\dots,A_p \in \R^{d\times d}$. This approach is denoted by VAR. The three models considered have a tuning parameter $p$ which specifies the lag order. For all three models,  the Bayesian Information Criterion (BIC) is used to automatically choose  the lag-order $p$, see among others Section 5.5 in \cite{BrockwellDavis1991}. 

The approaches NAR and LNAR make use of the underlying network structure. That means in order to compute $X_{n+h}^{h}$, the approaches NAR and LNAR require an observation or at least an estimate of the underlying network structure. Both cases are considered here, in the first case we forecast $Ad_{n},\dots,Ad_{n+h-1}$ based on the observations $Ad_1,\dots,Ad_{n-1}$, and in the second case we assume that $Ad_n,\dots,Ad_{n+h-1}$ is observed. In order to distinguish between these two cases, we denote the forecast of approach NAR based on an estimated network by NAR$(\widehat \Ad)$ and the forecast based on a known network structure by NAR$(\Ad)$. An analogue notation is used for LNAR.

All computations are done in \emph{R} \cite{R} using the additional packages \emph{tergm, BigVAR, markovchain} \cite{tergm_R,BigVAR,markovchain}.

In the first example, a network with $4$ vertices is considered. The adjacency matrix process $\Ad$ is a Markovian process and the edges are independent from each other. The process $\Ad$ is given by
\begin{align}
\left(P(Ad_{t;ij}=1| Ad_{t-1;ij}=1)\right)_{i,j=1,\dots,d}&=
 \begin{pmatrix}
      0.95 & 0.70 & 0.99 & 0 \\ 
 0 & 0.95 & 0.70 & 0 \\ 
 0.99 & 0.50 & 0.95 & 0.95 \\ 
 0.30 & 0 & 0 & 0.95
 \end{pmatrix}, \nonumber \\ 
 \left(
 P(Ad_{t;ij}=1| Ad_{t-1;ij}=0)\right)_{i,j=1,\dots,d}&=
 \begin{pmatrix}
 0.05 & 0.10 & 0.01 & 0 \\ 
 0 & 0.05 & 0.30 & 0 \\ 
 0.01 & 0.50 & 0.05 & 0.05 \\ 
 0.30 & 0 & 0 & 0.05 
 \end{pmatrix}. \label{example.1.network}
 \end{align}
 The process $\Xb$ is an NAR$(1)$ process and is given by 
 \begin{align}
 X_{t}= \left( \alpha
 \Mm Ad_{t-1}\right) X_{t-1} + \eps_{t}, t \in \Z,  \eps_{1} \sim \mathcal{N}\left( (-1 , 4,-9 ,16)^\top , I_4\right), \label{example.1.process}
 \end{align}
where $\alpha= \begin{pmatrix}
 0.25 & 0.7 & 0 & 0 \\ 
 0 & 0.25 & 0.7 & 0 \\ 
 0 & 0 & 0.25 & 0.7 \\ 
 0.7 & 0 & 0 & 0.25
 \end{pmatrix}$.
 
\begin{figure}[ht]
\begin{center}
\animategraphics[width=\textwidth,loop=TRUE,autoplay=TRUE]{1}{MarkovNetwork_d_4}{}{}
\includegraphics[width=\textwidth]{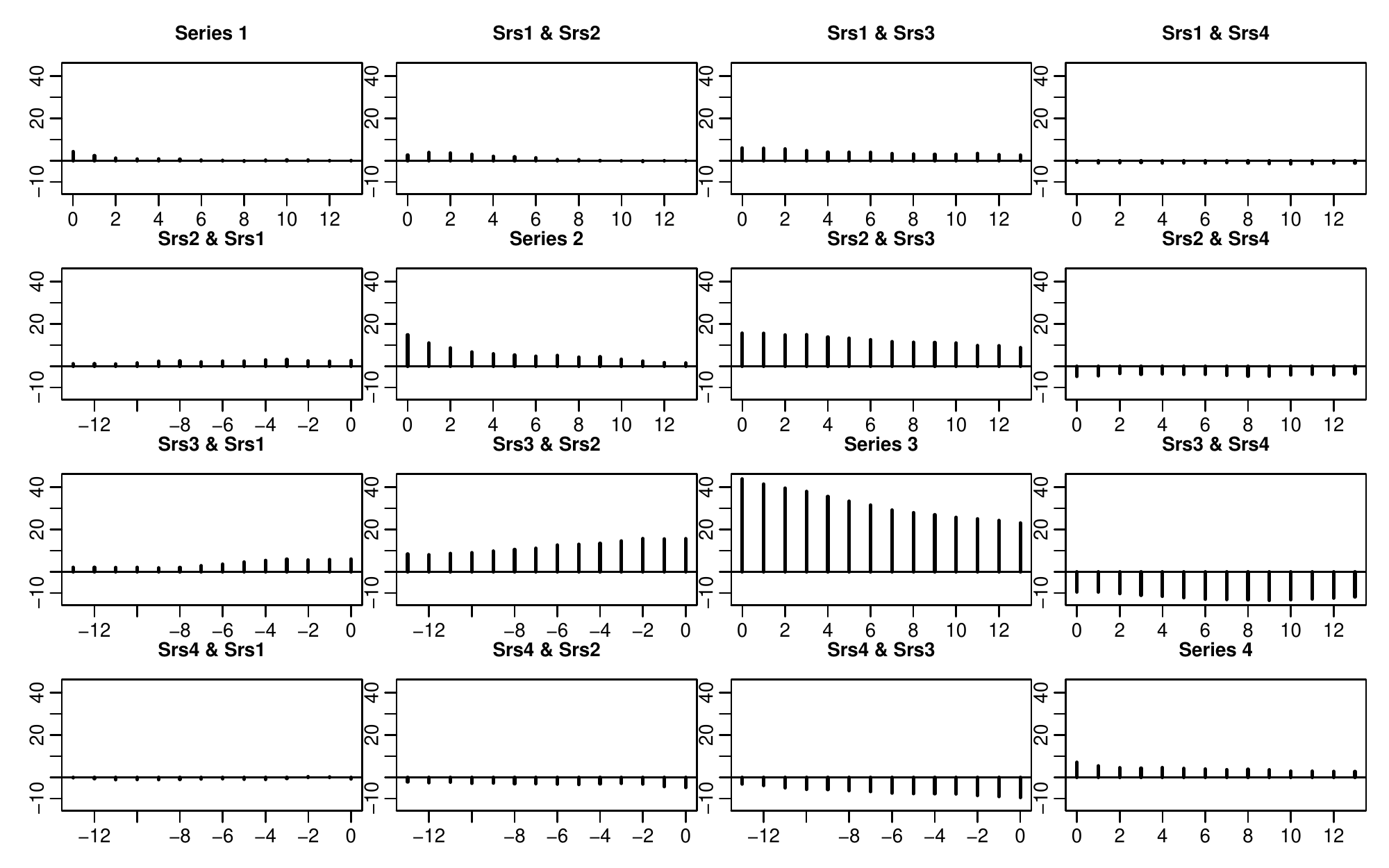} 
\end{center}
\caption[Realization of process $(\ref{example.1.process})$]{The upper figure presents a realization of the network of the example given by $(\ref{example.1.network})$ and a realization of the time series $X$ given by $(\ref{example.1.process})$. Red dots indicate the current time point. \textit{This figure contains animation only visible on screen.} The lower graphic presents the sample autocovariance function of $X$.} \label{example.1.fig.real}
\end{figure} 
 
A realization of the network, the time series, and the sample autocovariance function are displayed in Figure \ref{example.1.fig.real}. The edges $(3,1)$ and $(1,3)$ have a coefficient of $0$, hence, whether they are present or not, they do not influence the time series $\Xb$. 

The model structure of $\Ad$ is used to compute a forecast of $\Ad$. Thus, for each component of $\Ad$, a discrete Markov chain is fitted to $Ad_1,\dots,Ad_{n-1}$, and this Markov chain is then used to obtain a forecast  for $Ad_{n},\dots,Ad_{n+h-1}$. For this the \emph{R}-package \emph{markovchain} was used.

The mean squared errors for the forecast horizons $h=1,\dots,4$ are displayed in Table~\ref{table.example1}. Note that an optimal one-step ahead forecast for this process would possess a forecast error of $1$. This can be nearly achieved by NAR with a known network structure and moderate sample size. If the underlying network structure is unknown, NAR outperforms the other approaches for the forecast horizons up to $h=3$. For forecast horizons further ahead, VAR performs slightly better. This drop in performance for horizons further ahead is mainly caused by the estimate of the underlying network structure. The approach used here causes that $\hat Ad_{n+h}^{(h)}$ is identical for all horizons $h=1,\dots,4$. This estimate gets poorer, the larger $h$ is.

\begin{table}[t]
\centering
\begin{tabular}{l|l|cccc}
\hline
n   &                      & $h=1$ & $h=2$ & $h=3$ & $h=4$ \\ \hline
    & NAR$(\Ad)$           & 1.24  & 1.42  & 1.53  & 1.61  \\
    & NAR$(\widehat \Ad)$  & 1.47  & 2.80  & 3.67  & 4.14  \\
50  & LNAR$(\Ad)$          & 1.89  & 2.32  & 2.62  & 2.91  \\
    & LNAR$(\widehat \Ad)$ & 2.00  & 3.18  & 4.04  & 4.63  \\
    & VAR                  & 2.83  & 3.56  & 4.08  & 4.35  \\ \hline
    & NAR$(\Ad)$           & 1.07  & 1.19  & 1.23  & 1.25  \\
    & NAR$(\widehat \Ad)$  & 1.20  & 2.76  & 3.56  & 4.13  \\
100 & LNAR$(\Ad)$          & 1.78  & 2.18  & 2.35  & 2.47  \\
    & LNAR$(\widehat \Ad)$ & 1.83  & 3.08  & 3.86  & 4.49  \\
    & VAR                  & 2.65  & 3.40  & 3.80  & 4.08  \\ \hline
    & NAR$(\Ad)$           & 1.01  & 1.12  & 1.14  & 1.16  \\
    & NAR$(\widehat \Ad)$  & 1.11  & 2.63  & 3.55  & 4.09  \\
200 & LNAR$(\Ad)$          & 1.71  & 2.04  & 2.20  & 2.30  \\
    & LNAR$(\widehat \Ad)$ & 1.74  & 2.88  & 3.71  & 4.30  \\
    & VAR                  & 2.63  & 3.30  & 3.69  & 4.01  \\ \hline
    & NAR$(\Ad)$           & 1.01  & 1.11  & 1.14  & 1.15  \\
    & NAR$(\widehat \Ad)$  & 1.13  & 2.66  & 3.41  & 4.10  \\
500 & LNAR$(\Ad)$          & 1.69  & 2.00  & 2.14  & 2.23  \\
    & LNAR$(\widehat \Ad)$ & 1.73  & 2.91  & 3.60  & 4.32  \\
    & VAR                  & 2.60  & 3.27  & 3.60  & 3.92  \\ \hline
\end{tabular}
\caption{Mean squared error averaged over all components for process \eqref{example.1.process} and the forecast horizons $h=1,\dots,4$.} \label{table.example1}
\end{table}

\FloatBarrier
In the second example, a Separable Temporal Exponential Random Graph Model  (STERGM) is considered, see \citet{krivitsky2014separable} and also \cite{tergm_R} for the used \emph{R} package \emph{tergm}. The network is generated using \emph{simulate.stergm} of the R-package \emph{tergm} with dissolution coefficient $4$,  formation coefficient \linebreak $-\log\Big((d/5-1)(1+\exp(4))-1\Big)$, and a mean density of $5/d$. Networks of the sizes $d=10,33,100$, and $500$ are considered, and for each network size the sample sizes $n=100,200$, and $500$. For $d=500$, such a network has about $4000$ edge changes from $t=1$ to $t=100$. Let $Ad_t$ be the adjacency matrix of such a network at time $t$ and let $B_t=\operatorname{diag}((1/(\ind^\top  Ad_t e_i)_{i=1,\dots,d})$ be a diagonal matrix, where $1/0$ is defined as $0$. This defines the function $G(Ad_t)=Ad_t B_t$ and let $AdG_t=G(Ad_t)=Ad_t^\top B_t$. This means $e_i^\top AdG_t$ apportions equally the weight $1$  among the in-going edges to vertex $i$ at time $t$, and we have $\|AdG_t\|_\infty=\max_i\sum_{s=1}^d Ad_{t;si}/(\sum_{s=1}^d Ad_{t;si})\leq 1$. Then, the process $\Xb$ is given by the following LNAR(1) model
\begin{align}
    X_{t;r}=0.9(r/d) X_{t-1;r}+ 0.9 (d-r)/d AdG_{t-1} X_{t-1}+\eps_{t;r},
    \label{example.2.process}
\end{align}
where $\eps_t \sim \mathcal{N}(\mu,5\Sigma_\eps)$ and $\mu=(1(-1),2(-1)^2,3(-1)^3,\dots,d(-1)^d)^\top$ and $\Sigma_\eps$ is a banded matrix with ones on the diagonal and $0.25(-1)^{j+1},j=1,\dots,d-1$ on the first off diagonal. The function $G$ is considered as known, meaning $AdG_1,\dots,AdG_{n-1}$ is observed. 

Two approaches are used to obtain a forecast  for $Ad_{n},\dots,Ad_{n+h-1}$. The first approach fits a STERGM model to $Ad_1,\dots,Ad_{n-1}$ and then generates $\hat Ad_{n-1+h}^{(h)}$ by simulating the fitted model with $Ad_{n-1}$ as a starting value. This forecast is denoted as $\widehat{Ad}_1$. The second approach fits for each component independently a discrete Markov chain and uses this to forecast $\Ad$. The second approach is denoted as $\widehat{Ad}_2$.

Since a standard VAR model cannot be applied well to a high-dimensional setting, the VAR estimation is modified by adding sparsity constraints, meaning a coefficient $A_{j,rc}$ is set to zero for $j=1,\dots,p$ if $\sum_{t=1}^n e_r^\top AdG_t e_c=0$. This is motivated by the fact that in model \eqref{example.2.process} if $E e_r^\top AdG_t e_c=0$, then $X_{t-j;c},j=1,\dots,p$ do not directly influence $X_{t;r}$.
Furthermore, forecasts given by the \emph{R}-package \emph{BigVAR} are included as additional benchmarks, see \cite{BigVAR}. There the idea is that the underlying VAR model is sparse but the sparsity structure is unknown. It can be  estimated by using a LASSO approach, see among others \cite{basu2015,kock2015oracle}. The forecast obtained by a such a model of order $p$ is denoted by BigVAR$(p)$, where $p=1,2$.

The mean squared errors of the forecasts are displayed in Table~\ref{table.exampel2.10} to Table~\ref{table.exampel2.500}. Note that an optimal one-step ahead forecast would possess a forecast error of $5$. With a known network structure at hand, LNAR is nearly able to achieve such optimal results independently of the dimension. Is the future network unknown, then the forecasting performance of NAR and LNAR drop considerably. Especially for LNAR, the loss of performance due to an unknown network structure seems to increase with the dimension of the process. Of the two approaches used to forecast the network, a better performance is given by the network forecasting approach $\widehat{Ad}_1$ for NAR as well as LNAR. Given the network forecast $\widehat{Ad}_1$, both approaches still outperform all others namely VAR and BigVAR$(p), p=1,2$. As mentioned, VAR  uses the underlying network structure to set sparsity constraints such that the number of parameters can be reduced. In all settings considered, VAR outperforms BigVAR$(p), p=1,2$. This indicates that for this process the network induced sparsity constraints are more helpful than a free but unknown  sparsity setting as given in BigVAR. However, note that the amount of parameters which are estimated for the approaches NAR and VAR depend on the set $\mathcal{I}_n:=\{i,r=1,\dots,d : 1/n\sum_{t=1}^{n-1} e_r^\top|Ad_{t}|e_i>0\},$ and we have $|\mathcal{I}_n|\leq |\mathcal{I}_m|$ for $m\geq n$. This could explain why these two approaches do not gain immediately from an increasing sample size for larger dimensions.  

In order to get a better overview of the results, we set the MSE of the approach VAR as basing point, meaning all MSE values are divided by the corresponding VAR's MSE. An average over all sample sizes and forecast horizons leads to Table~\ref{table.example2.RELMSE}. These aggregated results support the argument that the approaches NAR and LNAR benefit from the underlying network structure and even when the future network structure is unknown and needs to be estimated itself, there is still a benefit. For higher dimensions, only LNAR performs well whereas the performance of all others drop dramatically, meaning that the network induced sparsity does not seem to be sufficient to obtain good estimation results in this setting.

\begin{table}[ht]
\centering
\begin{tabular}{l|rrrr}
  \hline
\multicolumn{1}{r|}{$d$} & 10 & 33 & 100 & 500 \\ 
  \hline
  NAR$(\Ad)$ & 0.35 & 0.33 & 0.44 & 0.51 \\ 
  NAR$(\widehat \Ad_1)$ & 0.71 & 0.74 & 0.78 & 0.77 \\ 
  NAR$(\widehat \Ad_2)$ & 0.84 & 0.93 & 0.95 & 0.96 \\ 
  LNAR$(\Ad)$ & 0.27 & 0.04 & 0.01 & $<0.01$ \\ 
  LNAR$(\widehat \Ad_1)$ & 0.66 & 0.60 & 0.56 & 0.53 \\ 
  LNAR$(\widehat \Ad_2)$ & 0.80 & 0.85 & 0.81 & 0.80 \\ 
  VAR & 1.00 & 1.00 & 1.00 & 1.00 \\ 
  BigVAR$(1)$ & 2.99 & 5.35 & 6.35 & 7.88 \\ 
  BigVAR$(2)$ & 3.24 & 5.71 & 6.68 & 8.09 \\ 
   \hline
   \end{tabular}
\caption{Relative mean squared error (basing point: VAR) for process \eqref{example.2.process} averaged over forecast horizons and sample sizes.} \label{table.example2.RELMSE}
\end{table}

\begin{table}[ht]
\setlength{\tabcolsep}{0.1cm}
\centering
\begin{tabular}{l|l|rrrr|rrrr}
  \hline
  && \multicolumn{4}{|c|}{$d=10$}& \multicolumn{4}{|c}{$d=33$}\\
$n$ & Model & $h=1$ & $h=2$ & $h=3$ & $h=4$ & $h=1$ & $h=2$ & $h=3$ & $h=4$ \\ 
\hline
  & NAR$(\Ad)$ & 6.9E+00 & 1.1E+01 & 1.5E+01 & 1.9E+01 & 3.4E+01 & 7.2E+01 & 1.1E+02 & 1.6E+02 \\ 
   & NAR$(\widehat \Ad_1)$ & 9.7E+00 & 1.8E+01 & 2.7E+01 & 3.6E+01 & 5.6E+01 & 1.3E+02 & 2.1E+02 & 2.9E+02 \\ 
   & NAR$(\widehat \Ad_2)$ & 1.2E+01 & 2.3E+01 & 3.2E+01 & 4.3E+01 & 7.5E+01 & 1.6E+02 & 2.5E+02 & 3.4E+02 \\ 
   & LNAR$(\Ad)$ & 5.1E+00 & 7.1E+00 & 8.5E+00 & 9.6E+00 & 5.3E+00 & 7.2E+00 & 8.5E+00 & 9.4E+00 \\ 
  100 & LNAR$(\widehat \Ad_1)$ & 8.3E+00 & 1.6E+01 & 2.3E+01 & 3.1E+01 & 4.1E+01 & 1.0E+02 & 1.7E+02 & 2.4E+02 \\ 
   & LNAR$(\widehat \Ad_2)$ & 1.1E+01 & 2.1E+01 & 2.9E+01 & 3.8E+01 & 6.7E+01 & 1.5E+02 & 2.3E+02 & 3.1E+02 \\ 
   & VAR & 1.3E+01 & 2.7E+01 & 4.1E+01 & 5.5E+01 & 6.9E+01 & 1.8E+02 & 3.2E+02 & 4.9E+02 \\
   & BigVAR$(1)$ & 4.4E+01 & 6.8E+01 & 8.8E+01 & 1.0E+02 & 5.4E+02 & 7.3E+02 & 8.9E+02 & 1.0E+03 \\ 
   & BigVAR$(2)$ & 4.8E+01 & 7.2E+01 & 9.1E+01 & 1.1E+02 & 5.7E+02 & 7.6E+02 & 9.2E+02 & 1.0E+03 \\ 
   \hline
 & NAR$(\Ad)$ & 5.9E+00 & 8.8E+00 & 1.1E+01 & 1.3E+01 & 2.9E+01 & 5.7E+01 & 8.5E+01 & 1.1E+02 \\ 
   & NAR$(\widehat \Ad_1)$ & 9.4E+00 & 1.8E+01 & 2.8E+01 & 3.7E+01 & 5.7E+01 & 1.2E+02 & 1.9E+02 & 2.6E+02 \\ 
   & NAR$(\widehat \Ad_2)$ & 1.1E+01 & 2.2E+01 & 3.2E+01 & 4.2E+01 & 7.6E+01 & 1.5E+02 & 2.4E+02 & 3.2E+02 \\ 
   & LNAR$(\Ad)$ & 5.2E+00 & 7.2E+00 & 8.3E+00 & 9.2E+00 & 5.1E+00 & 6.9E+00 & 8.0E+00 & 8.7E+00 \\ 
  200 & LNAR$(\widehat \Ad_1)$ & 8.9E+00 & 1.7E+01 & 2.6E+01 & 3.5E+01 & 4.2E+01 & 9.3E+01 & 1.5E+02 & 2.2E+02 \\ 
   & LNAR$(\widehat \Ad_2)$ & 1.1E+01 & 2.1E+01 & 3.0E+01 & 4.0E+01 & 6.6E+01 & 1.4E+02 & 2.2E+02 & 3.0E+02 \\ 
   & VAR & 1.3E+01 & 2.6E+01 & 4.0E+01 & 5.3E+01 & 6.2E+01 & 1.5E+02 & 2.7E+02 & 4.1E+02 \\ 
   & BigVAR$(1)$ & 5.2E+01 & 8.1E+01 & 1.0E+02 & 1.2E+02 & 6.9E+02 & 9.3E+02 & 1.1E+03 & 1.2E+03 \\ 
   & BigVAR$(2)$ & 6.0E+01 & 9.0E+01 & 1.1E+02 & 1.3E+02 & 7.5E+02 & 9.8E+02 & 1.2E+03 & 1.3E+03 \\ 
   \hline
 & NAR$(\Ad)$ & 5.1E+00 & 7.3E+00 & 8.8E+00 & 9.9E+00 & 1.9E+01 & 3.6E+01 & 5.2E+01 & 6.5E+01 \\ 
   & NAR$(\widehat \Ad_1)$ & 9.1E+00 & 1.7E+01 & 2.6E+01 & 3.3E+01 & 4.9E+01 & 1.1E+02 & 1.8E+02 & 2.5E+02 \\ 
   & NAR$(\widehat \Ad_2)$ & 1.1E+01 & 2.0E+01 & 2.9E+01 & 3.8E+01 & 6.5E+01 & 1.4E+02 & 2.2E+02 & 3.0E+02 \\ 
   & LNAR$(\Ad)$ & 5.0E+00 & 7.0E+00 & 8.3E+00 & 9.2E+00 & 5.0E+00 & 6.7E+00 & 7.8E+00 & 8.4E+00 \\ 
  500 & LNAR$(\widehat \Ad_1)$ & 9.0E+00 & 1.7E+01 & 2.6E+01 & 3.3E+01 & 3.9E+01 & 9.5E+01 & 1.6E+02 & 2.3E+02 \\ 
   & LNAR$(\widehat \Ad_2)$ & 1.1E+01 & 2.0E+01 & 2.9E+01 & 3.7E+01 & 5.7E+01 & 1.3E+02 & 2.1E+02 & 2.9E+02 \\ 
   & VAR & 1.2E+01 & 2.3E+01 & 3.5E+01 & 4.5E+01 & 5.4E+01 & 1.4E+02 & 2.5E+02 & 3.6E+02 \\ 
   & BigVAR$(1)$ & 5.0E+01 & 8.6E+01 & 1.1E+02 & 1.3E+02 & 5.1E+02 & 8.6E+02 & 1.1E+03 & 1.3E+03 \\ 
   & BigVAR$(2)$ & 5.6E+01 & 9.2E+01 & 1.2E+02 & 1.3E+02 & 5.7E+02 & 9.2E+02 & 1.2E+03 & 1.3E+03 \\ 
   \hline
\end{tabular}
\caption{Mean squared error for process \eqref{example.2.process} with dimension $d=10, 33$ and the forecast horizons $h=1,\dots,4$.} \label{table.exampel2.10}
\end{table}

\begin{table}[ht]
\setlength{\tabcolsep}{0.1cm}
\centering
\begin{tabular}{l|l|rrrr|rrrr}
  \hline
  && \multicolumn{4}{|c|}{$d=100$}& \multicolumn{4}{|c}{$d=500$}\\
$n$ & Model & $h=1$ & $h=2$ & $h=3$ & $h=4$ & $h=1$ & $h=2$ & $h=3$ & $h=4$ \\ 
 \hline
 & NAR$(\Ad)$ & 2.6E+02 & 5.7E+02 & 9.5E+02 & 1.4E+03 & 5.5E+03 & 1.2E+04 & 2.0E+04 & 2.7E+04 \\ 
   & NAR$(\widehat \Ad_1)$ & 4.2E+02 & 9.5E+02 & 1.5E+03 & 2.1E+03 & 9.5E+03 & 2.3E+04 & 3.7E+04 & 5.1E+04 \\ 
   & NAR$(\widehat \Ad_2)$ & 5.7E+02 & 1.2E+03 & 1.8E+03 & 2.5E+03 & 1.4E+04 & 3.0E+04 & 4.6E+04 & 6.1E+04 \\ 
   & LNAR$(\Ad)$ & 5.2E+00 & 7.3E+00 & 8.4E+00 & 9.3E+00 & 5.3E+00 & 7.2E+00 & 8.3E+00 & 9.1E+00 \\ 
  100 & LNAR$(\widehat \Ad_1)$ & 2.8E+02 & 7.3E+02 & 1.2E+03 & 1.8E+03 & 6.9E+03 & 1.8E+04 & 3.1E+04 & 4.4E+04 \\ 
   & LNAR$(\widehat \Ad_2)$ & 4.8E+02 & 1.1E+03 & 1.7E+03 & 2.4E+03 & 1.2E+04 & 2.8E+04 & 4.3E+04 & 5.9E+04 \\ 
   & VAR & 4.8E+02 & 1.4E+03 & 2.6E+03 & 4.3E+03 & 1.2E+04 & 3.5E+04 & 7.0E+04 & 1.2E+05 \\ 
   & BigVAR$(1)$ & 5.1E+03 & 6.4E+03 & 7.3E+03 & 8.1E+03 & 1.7E+05 & 1.9E+05 & 2.1E+05 & 2.2E+05 \\ 
   & BigVAR$(2)$ & 5.3E+03 & 6.4E+03 & 7.4E+03 & 8.2E+03 & 1.7E+05 & 1.9E+05 & 2.1E+05 & 2.2E+05 \\ 
   \hline
 & NAR$(\Ad)$ & 3.1E+02 & 6.5E+02 & 9.7E+02 & 1.3E+03 & 6.6E+03 & 1.5E+04 & 2.5E+04 & 3.7E+04 \\ 
   & NAR$(\widehat \Ad_1)$ & 5.0E+02 & 1.1E+03 & 1.7E+03 & 2.3E+03 & 1.1E+04 & 2.5E+04 & 4.0E+04 & 5.8E+04 \\ 
   & NAR$(\widehat \Ad_2)$ & 6.8E+02 & 1.4E+03 & 2.1E+03 & 2.8E+03 & 1.4E+04 & 3.0E+04 & 4.8E+04 & 6.8E+04 \\ 
   & LNAR$(\Ad)$ & 5.1E+00 & 6.9E+00 & 8.0E+00 & 8.5E+00 & 5.2E+00 & 6.9E+00 & 7.9E+00 & 8.7E+00 \\ 
  200 & LNAR$(\widehat \Ad_1)$ & 3.0E+02 & 7.3E+02 & 1.3E+03 & 1.9E+03 & 7.3E+03 & 1.8E+04 & 3.1E+04 & 4.5E+04 \\ 
   & LNAR$(\widehat \Ad_2)$ & 5.6E+02 & 1.2E+03 & 1.9E+03 & 2.5E+03 & 1.2E+04 & 2.6E+04 & 4.2E+04 & 5.8E+04 \\ 
   & VAR & 4.8E+02 & 1.3E+03 & 2.5E+03 & 3.9E+03 & 1.1E+04 & 3.3E+04 & 6.3E+04 & 1.0E+05 \\ 
   & BigVAR$(1)$ & 7.8E+03 & 9.4E+03 & 1.1E+04 & 1.1E+04 & 2.2E+05 & 2.4E+05 & 2.5E+05 & 2.7E+05 \\ 
   & BigVAR$(2)$ & 8.3E+03 & 9.8E+03 & 1.1E+04 & 1.2E+04 & 2.3E+05 & 2.5E+05 & 2.6E+05 & 2.7E+05 \\ 
   \hline
 & NAR$(\Ad)$ & 2.6E+02 & 5.7E+02 & 8.1E+02 & 1.0E+03 & 5.6E+03 & 1.4E+04 & 3.2E+04 & 1.1E+05 \\ 
   & NAR$(\widehat \Ad_1)$ & 4.6E+02 & 1.1E+03 & 1.7E+03 & 2.2E+03 & 7.9E+03 & 2.1E+04 & 4.3E+04 & 1.3E+05 \\ 
   & NAR$(\widehat \Ad_2)$ & 5.8E+02 & 1.3E+03 & 1.9E+03 & 2.5E+03 & 1.1E+04 & 2.8E+04 & 5.3E+04 & 1.4E+05 \\ 
   & LNAR$(\Ad)$ & 5.1E+00 & 6.7E+00 & 7.6E+00 & 8.2E+00 & 5.0E+00 & 6.8E+00 & 7.6E+00 & 8.3E+00 \\ 
  500 & LNAR$(\widehat \Ad_1)$ & 3.0E+02 & 8.1E+02 & 1.3E+03 & 1.8E+03 & 6.6E+03 & 1.7E+04 & 2.9E+04 & 4.3E+04 \\ 
   & LNAR$(\widehat \Ad_2)$ & 4.7E+02 & 1.1E+03 & 1.7E+03 & 2.3E+03 & 1.2E+04 & 2.7E+04 & 4.2E+04 & 5.7E+04 \\ 
   & VAR & 4.3E+02 & 1.2E+03 & 2.2E+03 & 3.2E+03 & 1.0E+04 & 2.9E+04 & 5.3E+04 & 8.4E+04 \\ 
   & BigVAR$(1)$ & 5.6E+03 & 8.1E+03 & 9.8E+03 & 1.1E+04 & 2.0E+05 & 2.5E+05 & 2.7E+05 & 2.8E+05 \\ 
   & BigVAR$(2)$ & 6.2E+03 & 8.6E+03 & 1.0E+04 & 1.1E+04 & 2.1E+05 & 2.4E+05 & 2.7E+05 & 2.8E+05 \\ 
   \hline
\end{tabular}
\caption{Mean squared error for process \eqref{example.2.process} with dimension $d=100, 500$ and the forecast horizons $h=1,\dots,4$.} \label{table.exampel2.500}
\end{table}

\FloatBarrier
\section{Real Data Example} \label{net.real.data}
In this section, we investigate further the example in which the actors are economies, their gross domestic product is the attribute of interest and their trade volume defines the underlying network. To elaborate, we consider the data set of \cite{mohaddes2018compilation}. This data set contains economic data of $33$ economies in the time period from 1980-2016. The $33$ economies cover more than $90\%$ of world GDP, see Table~\ref{table.GDP.econ} for a list of included economies. The economies are considered as actors, and the relationship between these actors is given by the IMF (International Monetary Fund) Direction of Trade statistics, see \emph{data.imf.org/DOT} and also the trade matrix in \cite{mohaddes2018compilation}. For time $t$, the connection from actors $i$ to actor $j$ given by $e_j^\top Ad_t^\top e_i$ is defined as the sum of exports and imports between actor $i$ and $j$ at time $t$ divided by the sum of all exports and imports of actor $j$ at time $t$. The data set considered contains for each economy attributes such as real GDP (log transform), inflation rate, short/long-term interest rate. Note that these attributes are given quarterly whereas the trade relations are only given annually. We assume here that the trade relations do not change within a year and perform the analysis on the quarterly sampling level. The focus here is on the attribute real GDP, and based on the data from 1980Q1-2014Q4 the goal is to forecast the GDP for the period 2015Q1-2016Q4. The indices $1,\dots,n$ denote the time period 1980Q1-2014Q4 and $n+1,\dots,n+8$ the time period 2015Q1-2016Q4. To perform a forecast, we use the models presented here, namely NAR given by \eqref{GNAR.eq} and LNAR given by \eqref{eq.LNAR}, and include a VAR model as a benchmark. It is a solid benchmark, since  \citet{marcellino2007comparison} compared a VAR model GDP forecast with various nonlinear alternatives and pointed out that even though a VAR model is a ``simple'' linear model, it can hardly be beaten if it is carefully specified.
Let $\{Y_t \}$ be the real GDP (log transform) of the $33$ economies. Unit root tests applied to $\{Y_t \}$ suggest that real GDP itself may not be stationary. We follow here the economic literature, see among others \cite{marcellino2007comparison}, and model instead the GDP growth rate given by $X_t=Y_t-Y_{t-1}$. This transformation can be inverted, and we obtain a forecast for $\{Y_t\}$ by $\hat Y_{n+h}^{(1)}=Y_{n}+\sum_{s=1}^h\hat X_{n+s}^{(s)}$, where $\hat X_{n+s}^{(s)},s=1,\dots,h$ denote forecasts of $\{X_t\}$. 
The three models considered have a tuning parameter $p$ which specifies the lag order. For all three models,  the Bayesian Information Criterion (BIC) is used to choose automatically the lag-order $p$. The models NAR and LNAR require a forecast of the underlying network. A simple approach for this is used, namely $\hat Ad_{n+h}^{(h)}=Ad_n, h=1,\dots,8$ is used as a forecast.  We obtain for a forecast $\hat Y_{n+h}^{(h)}$ the forecast error $E_{n+h}^{(h)}:=Y_{n+h}-Y_{n+h}^{(h)}$. The squared error and the absolute error is used to measure the forecast performance. In Table~\ref{table.TOTAL.error}, the sum of all squared and absolute errors is displayed, meaning $\sum_{h=1}^8 \|E_{n+h}^{(h)}\|_2^2$ and $\sum_{h=1}^8 \|E_{n+h}^{(h)}\|_1$. Over all forecast horizons and economies, the additional network structure improves the forecast, and NAR performs best. NAR's forecast error is $25\%$ for the squared error and $15\%$ for the absolute error smaller than VAR's forecast error.
\begin{table}[ht]
\centering
\begin{tabular}{rrrr}
  \hline
 & VAR & LNAR & NAR \\ 
  \hline
Square error & 0.23 & 0.21 & 0.17 \\ 
  Absolute error & 5.21 & 4.87 & 4.39 \\ 
   \hline
\end{tabular}
\caption{Sum of the squared error and absolute error, respectively, over all $33$ economies and all $8$ forecast horizons for the GDP forecast.} \label{table.TOTAL.error}
\end{table}

Table~\ref{table.GDP.econ} breaks the forecast error down  into economies, and Table~\ref{table.GDP.horizon} breaks it down into forecast horizons.  The performance gap between NAR and LNAR is small for small $h$, and it increases with increasing $h$. Since for $h=8$ LNAR and VAR almost have the same performance and NAR outperforms both, it seems that LNAR performs worse with increasing $h$. Taking a closer look at Table~\ref{table.GDP.econ}, we cannot identify a clear winner. For the $33$ listed economies VAR performs independently of the used error measure 11 times best, LNAR 6 times and NAR 16 times.

To sum up, the trade network delivers useful information for the GDP forecast. The models presented in this paper are able to benefit from these additional information such that they can outperform the VAR approach. Note that an NAR$(p)$ model possesses the same amount of parameters as a VAR$(p)$ model. Thus, the additional information can be used without estimating additional parameters.

\begin{table}[ht]
\centering
\begin{tabular}{l|rrr|rrr}
  \hline
  & \multicolumn{3}{|c|}{Squared error ($\times$100) } & \multicolumn{3}{c}{Absolute error}\\
 & VAR & LNAR & NAR & VAR & LNAR & NAR \\ 
  \hline
USA & 0.114 & 0.147 & 0.033 & 0.079 & 0.095 & 0.044 \\ 
  UNITED KINGDOM & 0.057 & 0.107 & 0.079 & 0.062 & 0.088 & 0.067 \\ 
  AUSTRIA & 0.133 & 0.211 & 0.404 & 0.096 & 0.124 & 0.165 \\ 
  BELGIUM & 0.094 & 0.097 & 0.085 & 0.079 & 0.077 & 0.077 \\ 
  FRANCE & 0.087 & 0.039 & 0.004 & 0.074 & 0.050 & 0.016 \\ 
  GERMANY & 0.032 & 0.056 & 0.329 & 0.048 & 0.062 & 0.148 \\ 
  ITALY & 0.045 & 0.004 & 0.038 & 0.054 & 0.016 & 0.046 \\ 
  NETHERLANDS & 0.010 & 0.035 & 0.013 & 0.024 & 0.048 & 0.026 \\ 
  NORWAY & 0.701 & 0.508 & 0.155 & 0.217 & 0.179 & 0.096 \\ 
  SWEDEN & 0.185 & 0.250 & 0.134 & 0.106 & 0.127 & 0.090 \\ 
  SWITZERLAND & 0.188 & 0.268 & 0.477 & 0.118 & 0.140 & 0.184 \\ 
  CANADA & 0.401 & 0.449 & 0.146 & 0.168 & 0.179 & 0.104 \\ 
  JAPAN & 0.101 & 0.013 & 0.085 & 0.078 & 0.030 & 0.076 \\ 
  CHINA & 1.099 & 0.905 & 0.321 & 0.257 & 0.235 & 0.131 \\ 
  FINLAND & 0.268 & 0.153 & 0.036 & 0.142 & 0.110 & 0.047 \\ 
  SPAIN & 0.086 & 0.043 & 0.939 & 0.078 & 0.054 & 0.246 \\ 
  TURKEY & 0.326 & 0.450 & 1.330 & 0.123 & 0.182 & 0.248 \\ 
  AUSTRALIA & 0.098 & 0.069 & 0.038 & 0.074 & 0.059 & 0.046 \\ 
  NEW ZEALAND & 0.055 & 0.095 & 0.441 & 0.056 & 0.073 & 0.157 \\ 
  SOUTH AFRICA & 0.711 & 1.063 & 0.059 & 0.210 & 0.254 & 0.064 \\ 
  ARGENTINA & 0.486 & 0.782 & 0.981 & 0.176 & 0.207 & 0.233 \\ 
  BRAZIL & 6.896 & 8.437 & 7.244 & 0.683 & 0.748 & 0.700 \\ 
  CHILE & 0.977 & 1.180 & 0.382 & 0.240 & 0.251 & 0.149 \\ 
  MEXICO & 0.049 & 0.019 & 0.184 & 0.058 & 0.034 & 0.111 \\ 
  PERU & 0.057 & 0.017 & 1.143 & 0.058 & 0.028 & 0.271 \\ 
  SAUDI ARABIA & 0.286 & 0.243 & 0.357 & 0.145 & 0.116 & 0.164 \\ 
  INDIA & 0.060 & 0.093 & 0.085 & 0.064 & 0.079 & 0.074 \\ 
  INDONESIA & 0.006 & 0.059 & 0.037 & 0.020 & 0.061 & 0.048 \\ 
  KOREA & 1.700 & 0.256 & 0.008 & 0.334 & 0.129 & 0.019 \\ 
  MALAYSIA & 0.510 & 0.404 & 0.122 & 0.184 & 0.160 & 0.071 \\ 
  PHILIPPINES & 1.442 & 0.868 & 0.424 & 0.302 & 0.225 & 0.154 \\ 
  SINGAPORE & 4.251 & 2.714 & 0.315 & 0.544 & 0.427 & 0.156 \\ 
  THAILAND & 1.020 & 0.767 & 0.367 & 0.260 & 0.228 & 0.163 \\ 
   \hline
   \end{tabular}
   \caption{For a given economy the sum of squared error and absolute error, respectively, of the GDP forecast for the entire forecast period 2015Q1-2016Q4} \label{table.GDP.econ}
\end{table}

\begin{table}[ht]
\centering
\begin{tabular}{l|rrr|rrr}
  \hline
  & \multicolumn{3}{|c|}{Squared error ($\times$100) } & \multicolumn{3}{c}{Absolute error}\\
$h$ & VAR & LNAR & NAR & VAR & LNAR & NAR \\ 
  \hline
1 & 0.23 & 0.17 & 0.15 & 0.20 & 0.18 & 0.17 \\ 
  2 & 0.95 & 0.69 & 0.68 & 0.40 & 0.35 & 0.34 \\ 
  3 & 1.31 & 1.05 & 0.98 & 0.48 & 0.41 & 0.40 \\ 
  4 & 2.18 & 1.87 & 1.49 & 0.63 & 0.55 & 0.50 \\ 
  5 & 2.87 & 2.58 & 2.07 & 0.70 & 0.66 & 0.58 \\ 
  6 & 3.81 & 3.69 & 2.75 & 0.82 & 0.80 & 0.68 \\ 
  7 & 5.27 & 4.95 & 4.40 & 0.98 & 0.93 & 0.85 \\ 
  8 & 5.90 & 5.79 & 4.28 & 1.00 & 0.99 & 0.87 \\ 
   \hline
\end{tabular}
\caption{For a given forecast horizon the sum of the squared error and absolute error, respectively, over all $33$ economies.} \label{table.GDP.horizon}
\end{table}

\FloatBarrier
\section{Conclusions}
This paper models dynamic attributes of the vertices of a dynamic network. The attributes are modeled such that the underlying network structure can influence the attributes and vice versa. A linear time series framework is adopted and network linear processes and network autoregressive processes were defined. This framework gives flexibility in the sense that the attributes and the underlying network can be modeled separately. The physical dependence framework is used to quantify the dependency structure of the underlying network such that this framework becomes feasible and statistical results can be derived in a low- and high-dimensional setting.
These results can be used to do forecasting, and, as can be seen in the numerical examples as well as in the real data example, the benefit of using the additional structure can be quite large.
\FloatBarrier

{\bf Acknowledgments.}  The author is grateful to the editor, an associate editor and one referee for their valuable and insightful comments that led to a considerably improved manuscript. The research of the author was supported by the Research Center (SFB) 884 ``Political Economy of Reforms''(Project B6), funded by the German Research Foundation
(DFG). Furthermore, the author acknowledges support by the state of Baden-W{\"u}rttemberg through bwHPC.

\section{Proofs} \label{net.proofs}
\begin{proof}[Proof of Lemma~\ref{lem.solution}]
Under Assumption~\ref{ass.stat}a), we have that $\|G_j(\cdot)\|_{\max}\leq 1$ implies $\|\widetilde{G(\cdot)}\|_{\max}=1$, where $\widetilde{G(\cdot)}$ denotes the corresponding quantity of the stacked processes. Thus, for $j \in \N$ and $t \in \Z$, we have $\|\prod_{s=1}^j ( \tilde A \Mm \widetilde{G(Ad}_{t-s})\|_2 \leq \| |\tilde A|^j\|_2$. 
The condition $\det(I-\sum_{j=1}^p |A_j| z^j)\not = 0$ for all $|z|\leq 1$ implies $\det(I- \tilde A z^j)\not = 0$ for all $|z|\leq1$, which gives  component-wise summability of the sequence $\sum_{j=0}^\infty |A|^j$, see Appendix A.6 and A.9 in \cite{luetkepohl2007new}.
Under Assumption~\ref{ass.stat}b), we have also by the results of A.9 in \cite{luetkepohl2007new} that 
$\sum_{j=0}^\infty \tilde A \Mm \tilde G(\cdot)$ is absolutely component-wise summable. 

Let $\mathds{X}_t=(\tilde A \Mm \widetilde{G(Ad)}_{t-1}) \mathds{X}_{t-1}+(e_1 \otimes I_d) \eps_t$ be the stacked NAR$(1)$ process, where \eqref{eq.NAR.MAinfty} takes the form $\mathds{X}_t=\sum_{j=0}^\infty \prod_{s=1}^j (\tilde A \Mm G(\widetilde{Ad}_{t-s})) ( e_1 \otimes I_d) \eps_{t-j}$.  This is obviously a solution of the recursion equality. Given \eqref{eq.NAR.MAinfty}, the representation of the autocovariance function follows directly by taking into account that $\{\eps_t\}$ is an i.i.d. sequence and $\{\eps_s, s>t\}$ and $\{\Ad_s, s\leq t\}$ are independent for all $t$.
\end{proof}

\begin{proof}[Proof of Lemma \ref{lem1.DSLP}]
Condition (\ref{ass.5a}) and (\ref{ass.5b}) gives the existence of the $L_2$-Limit of $X_t$, so that it can be written as $X_{t} = \sum_{j=0}^\infty B_{t,j} \eps_{t-j}$. We have $B_{t,j}=f_j(Ad_{t-1},\dots,Ad_{t-j})$, and $\{ \eps_{t}, t \in \Z\}$ is i.i.d and independent to the stationary process $\Ad$. Thus, $\{\eps_{t}, t \in \Z\}$ and $(\veco(B_{t,j},j\in \N))_{t\in \Z}$ are independent.  We have  $\mu_x = \sum_{j=0}^\infty E B_{0,j} \mu$, and for the autocovariance function
\begin{align*}
\Gamma_X(h) &= \sum_{j=0}^\infty \sum_{s=0}^\infty  \left( E \left( B_{t+h,j} \eps_{t+h-j} \eps_{t-s}^\top B_{t,s}^\top \right) - E \left( B_{t+h,j} \mu \mu^\top B_{t,s} \right) \right. + \\  &\quad\left.
  E \left( B_{t+h,j} \mu \mu^\top B_{t,s} \right) - E  (B_{t+h,j}) \mu \mu^\top E (B_{t,s}^\top) \right) \\
&=\sum_{s=0}^\infty E \left(B_{h,s+h} \Sigma_\eps B_{0,s}^\top\right) +\sum_{j=0}^\infty \sum_{s=0}^\infty \cov\left( B_{h,j} \mu, B_{0,s} \mu\right), h \geq 0.
\end{align*}
\end{proof}

\begin{proof}[Proof of Theorem~\ref{thm.cons.NAR}]
First note that we have for $n$ large enough that  \linebreak 
$\sum_{t=p}^{n-1} \Yr_t (\Yr_t)^\top - 1/(n-p) \sum_{t_1,t_2=p}^{n-1} (\Yr_{t_1})  (\Yr_{t_2})^\top$ is invertible due to $\Gamma_{\Yr}(0)^{-1}$ is positive definite. Then, we insert $X_{t;r}= (\Yr_{t-1})^\top w_r+\eps_{t;r}$ in \eqref{eq.sys.est} and obtain
\begin{align*}
    \hat w_r =& w_r + [\sum_{t=p}^{n-1} \Yr_t (\Yr_t)^\top - \frac{1}{n-p} \sum_{t_1,t_2=p}^{n-1} (\Yr_{t_1})  (\Yr_{t_2})^\top]^{-1} \\
    & [\sum_{t=p+1}^n \eps_{t;r} \Yr_{t-1} - \frac{1}{n-p} \sum_{t_1,t_2=p+1}^n \Yr_{t_1-1} \eps_{t_2;r} ]
\end{align*}
and
$\hat \mu_r = 1/(n-p) \sum_{t=p+1}^n (\Yr_{t-1})^\top (w_r - \hat w_r)+1/(n-p)\sum_{t=p+1}^n \eps_{t;r}$. The $\sqrt{n}$-consistency of the estimators follows then by Assumption~\ref{ass.conv}. 
Since $\sqrt n/(n-p)\sum_{t=p+1}^n(\Yr_{t-1})^\top (w_r-\hat w_r)= \mu_{\Yr}(w_r-\hat w_r) + o_P(1)$, we have the second assertion by Assumption~\ref{ass.norm}.
\end{proof}

\begin{proof}[Proof of Lemma~\ref{lem.physical}]
To simplify notation, let $AdG_t=\widetilde{G(Ad}_t)$.
Since Assumption~\ref{ass.stat}a) gives a causal representation, see Lemma~\ref{lem.solution}, we have for component $r$ that  $X_{j;r}=\sum_{s=0}^\infty  e_r^\top B_{j,s} \eps_{j-s}=H_r(\Xi_j, \Xi_{j-1},\dots)$ for some measurable function $H_r$. Note that $B_{j,0} \equiv I_p$ and $\|B_{j,s}\|\leq \||\tilde A|^s\|_2$ for all $j$. Denote by $^*$ a coupled version with $\Xi_0$ being replaced by an i.i.d. copy $\Xi_0^\prime$. Then, $\delta_q(\{X_{t;r}, t\in\Z\},j)=\|X_{j;r}-H_r(\Xi_j,\Xi_{j-1},\dots,\Xi_1,\Xi_0^\prime,\Xi_{-1},\Xi_{-2},\dots)\|_{E,q}=\|X_{j;r}-X_{j;r}^*\|_{E,q}$.   We have by triangular inequality and Cauchy-Schwarz for $j\geq1$
\begin{align*}
\delta_q&(\{X_{t;r}\},j)=\|X_{j;r}-X_{j;r}^*\|_{E,q}
=\| \sum_{s=1}^\infty e_r^\top(B_{j,s}-B_{j,s}^*) \eps_{j-s} + e_r B_{j,j}^*(\eps_0-\eps_0^\prime)\|_{E,q} \\
&\leq \sum_{s=1}^\infty \sum_{i=1}^d \| (B_{j,s;ri}-B_{j,s;ri}^*) \eps_{j-s;i}\|_{E,q}+\|e_r B_{j,j}^*(\eps_0-\eps_0^\prime)\|_{E,q}\\
&\leq \sum_{s=1}^\infty \sum_{i=1}^d \| (B_{j,s;ri}-B_{j,s;ri}^*)\|_{E,2q}\| \eps_{j-s;i}\|_{E,2q}+\max_{i} \|E_{0;i}\|_{E,2q} \sqrt{d} \||\tilde A|^j\|_2.
\end{align*}
Let $D_k=AdG_{j-k}-AdG_{j-k}^*$. Since $\Ad_G$ is $2q$-stable, we have that \linebreak $\max_{i,s}\|e_i^\top D_k e_s\|_{E,2q}\leq \delta_{2q}(\Ad_G,j-k)$. Note that $\Ad_G$ is a causal process and $AdG_t-AdG^*_t=0$ for all $t<0$.
Furthermore, we have by Assumption~\ref{ass.stat}a)
\begin{align*}
 \| e_r^\top (B_{j,s}-B_{j,s}^*) \ind\|_{E,2q}=&
\| e_r^\top \Big( \sum_{k=1}^s \prod_{r=1}^{k-1} (\tilde A \Mm AdG_{j-k}^*) (\tilde A \Mm (AdG_{j-k}-AdG_{j-k}^*)) \\
&\quad \times \prod_{r=k+1}^s (\tilde A \Mm AdG_{j-r})\Big) \ind \|_{E,q}
\\&\leq \sqrt {d} \| |\tilde A|^s\|_2 \sum_{k=1}^{\min(j,s)} \delta_{2q}(\Ad_G,j-k).   
\end{align*}
With this, we have further
$$
\delta_q(\{X_{t}\},j)\leq \max_{i} \|\eps_{0;i}\|_{E,2q} \sqrt{d} \Big[ \sum_{s=1}^\infty \| |\tilde A|^s\|_2 \sum_{k=1}^{\min(j,s)} \delta_{2q}(\Ad_G,j-k)+\| |\tilde A|^j\|_2\Big].
$$
Similarly, we obtain 
$\max_r \|X_{j;r}\|_{E,q}\leq \max_i \|\eps_{0,i}\|_{E,2q} \sqrt d  \sum_{s=0}^\infty \| | \tilde A\|^s\|_2<\infty.$
Furthermore, let $C=\sum_{j=1}^\infty \||\tilde A|^j\|_2$. Since the components are absolutely summable, see Lemma~\ref{lem.solution}, we have $C<\infty$. This gives us
\begin{align*}
    \sum_{j=1}^\infty& \delta_q(\{X_{t},t\in \Z\},j)\leq \max_{i} \|\eps_{0;i}\|_{E,2q} \sqrt{d} (\sum_{j=1}^\infty \sum_{s=1}^\infty \| |\tilde A|^s \|_2 \!\!\!\!\!\! \sum_{k=1}^{\min(j,s)} \!\!\!\! \delta_{2q}(\Ad_G,j-k) + C) \\
    &\leq \max_{i} \|\eps_{0;i}\|_{E,2q} \sqrt{d} (\sum_{j=0}^\infty \delta_{2q}(\Ad_G,j) \sum_{s=0}^\infty \| |\tilde A|^s \|_2 \sum_{k=0}^\infty  |\tilde A|^k \|_2  + C) < \infty.
\end{align*}
Hence, the assertion that $\Xb$ is $q$-stable follows.

Let $\tilde k=k+(s-1)d, k=1,\dots,d, s=1,\dots,p$. Since $\Ad_G$ is stationary and $2q$-stable, we have $\ind(1/n \sum_{t=p+1}^n e_r^\top |G_s(Ad_{t-s})| e_k>0)=\ind( E(e_k^\top |G_s(Ad_{-s})|e_r>0)>0)+o_P(1)$, where $\ind(\cdot)$ is the indicator function, meaning $\ind(x)=1$ if $x$ is true and zero otherwise. This implies, as $n\to \infty$, $|\mathcal{I}^{(r)}|\overset{P}{\to} |\mathcal{I}^{(r)}_E|<\infty.$ Consequently, for $n$ large enough, we have $\mathcal{I}^{(r)}_E=\mathcal{I}^{(r)}$. Suppose in the following that this hold. 

A $q$-stable $\Xb$ implies that $\{(Y_{t-1}^{(r)}=I_{dp;I^{(r)}}( (e_r^\top G_1(Ad_{t-1})) \Mm X_{t-1},\dots,\linebreak (e_r^\top G_p(Ad_{t-p})) \Mm X_{t-p})^\top)_{r=1,\dots,d}, t \in \Z\}$ is $q$-stable. To see this, let $\tilde k=k+(s-1)d, \tilde k \in \mathcal{I}^{(r)}$.
Then, $e_{\tilde k}^\top Y_{t-1}^{(r)}=X_{t-s;k}e_r^\top G_s(Ad_{t-j}) e_k$ and $\| e_{\tilde k} (Y_j^{(r)}-Y_j^{(r)*})\|_{E,q}=\| (X_{j-s;k}-X_{j_s;k}^*) e_r^\top G_s(Ad_{j-s})e_k+X_{j-s;k}^*(e_r^\top G_s(Ad_{j-s})e_k -\linebreak e_r^\top G_s(Ad_{j-s}^*)e_k)\|_{E,q}\leq \delta_q(\Xb;k-j)+\delta(\Ad_G;k-j)$.
Since $\{X_t=\sum_{j=0}^\infty B_{t,j} \eps_{t-j}\}$ is a causal process, and $\{\eps_s, s> t\}$ and $\{Ad_s, s\leq t\}$ are independent for all $t$, we have that $Y_t^{(r)}$ and $\eps_{t+1}$ are independent for all $t$ and $r$. Hence, the i.i.d. structure of $\{\eps_t\}$ implies that $\{(Y_t^{(r)}-\mu_\Yr)(\eps_{t+1;r}-\mu_r),t\in \Z\}$ is $q$-stable and centered. This ensures that Theorem 3 in \cite{wu2011asymptotic} can be applied, see also Proposition 3 in \cite{wu2007strong}. This gives us that Assumption~\ref{ass.conv}(2.) follows.

Let $\tilde k_1=k_1+(s_1-1)d, \tilde k_2=k_2+(s_1-1)d \in I^{(r)}$. Then, for all $t\in \Z$
$\cov(
e_{\tilde k_1}^\top((e_r^\top G_1(Ad_{t-1})) \Mm X_{t-1},\dots,(e_r^\top G_p(Ad_{t-p})) \Mm X_{t-p})^\top,e_{\tilde k_2}^\top((e_r^\top G_1(Ad_{t-1})) \linebreak \Mm X_{t-1},\dots,(e_r^\top G_p(Ad_{t-p})) \Mm X_{t-p})^\top)=\cov(X_{-s_1;k_1} e_{r}^\top G_{s_1}(Ad_{-s_1}) e_{k_1},X_{-s_2;k_2} \linebreak e_{r}^\top G_{s_2}(Ad_{-s_2}) e_{k_2})=: e_{\tilde k_1}^\top I_{dp;\mathcal{I}^{(r)}}^\top \Gamma_{\Yr}(0) I_{dp;\mathcal{I}^{(r)}} e_{\tilde k_2}$. Since $\var(\eps_1)=\Sigma_\eps$ is positive definite, which implies $\var (X_1)$ is positive definite, and $E e_r^\top |G_s(Ad_1)|e_k>0$, $\Gamma_{Yr}(0)$ is positive definite. Hence, the $q$-stability and Theorem 3 in \cite{wu2011asymptotic} gives Assumption~\ref{ass.conv}(1.). Furthermore, Theorem 3 in \cite{wu2011asymptotic} in connection with the Crámer-Wold device gives Assumption~\ref{ass.norm}.
\end{proof}

\begin{proof}[Proof of Lemma~\ref{lem.stat.LNAR}]
We have 
$$\rho(|\tilde A \Mm \tilde G(\cdot)|)\leq \||\tilde A \Mm \tilde G(\cdot)|\|_\infty=\max(1,\max_{r=1,\dots,p} \sum_{j=1}^p |\alpha_{j,r}| + | \beta_{j,r}| \|G_j(\cdot)\|_\infty) = 1.$$
Suppose $v=(v_1^\top,\dots,v_p^\top)^\top \in \C^{dp}, v_k \in \C^d $ is an eigenvector of $|\tilde A \Mm \tilde G(\cdot)|$ to eigenvalue $\lambda$. Due to the special structure of $\tilde A \Mm \tilde G(\cdot)$ we have the following equations
$\sum_{j=1}^p |A_{j,\alpha,\beta} \Mm G_j(\cdot)| v_j = \lambda v_1$ and $v_k=\lambda v_{k+1}, k=1,\dots,p-1$. Following the proof of Theorem 1 in \cite{knight2019generalised} we obtain 
$\sum_{j=1}^p |A_{j,\alpha,\beta} \Mm G_j(\cdot)|\lambda^{-j} v_p =  v_p$. We have $0=\|\sum_{j=1}^p |A_{j,\alpha,\beta} \Mm G_j(\cdot)|\lambda^{-j} v_p\|_1-\|v_p\|_1\leq (|\lambda|^{-p} C_\lambda-1) \|v_p\|_1$. Hence, we obtain for an eigenvalue $|\lambda|^{p}\leq C_\lambda$. Since $C_\lambda<1$ this implies $\rho(|\tilde A \Mm \tilde G(\cdot)|)<1$. Furthermore, note that this   implies Assumption~\ref{ass.stat}b), which gives stationarity by Lemma~\ref{lem.solution}.
\end{proof}

\begin{proof}[Proof of Lemma~\ref{lem.physical.LNAR}]
To simplify notation, let $AdG_t=\widetilde{G(Ad}_t)$. Similarly as in the proof of Theorem~\ref{lem.physical}, we have
\begin{align*}
    \delta_q&(\{\max _r X_{\cdot;r}\},j)=\|\max_r X_{j;r}-H_r(\Xi_j,\Xi_{j-1},\dots,\Xi_1,\Xi_0^\prime,\Xi_{-1},\Xi_{-2},\dots)\|_{E,q}\\
    &\leq \sum_{s=1}^\infty \|\max_r e_r ^\top (B_{j,s}-B_{j,s}^*) \eps_{j-s}\|_{E,q}+\max_{i} \|\eps_{0,i}\|_{E,2q} \|\max_r e_r^\top B_{j,j}^*\ind \|_{E,2q}.
\end{align*}
Note that for $j \in Z,s \in \N$, we have  $B_{j,s}=B_{j,s-1}B_{j-s,1}=B_{j-1,1} B_{j-1,s-1}$. This gives us $B_{j,s}-B_{j,s}^*=\sum_{k=1}^s B_{j,k-1} ( B_{j-k,1} - B_{j-k,1}^*) B_{j-k,s-k}^*$. Note that $B_{k,1}=B_{k,1}^*$ for all $k<0$. Hence,  
\begin{align*}
    \|\max_r& e_r ^\top (B_{j,s}-B_{j,s}^*) \eps_{j-s}\|_{E,q}    \leq \max_{i} \| \eps_{0,i}\|_{E,2q}  \sum_{k=1}^{\min(s,j)} \| |\tilde A_{\alpha,\beta} \Mm \tilde G(\cdot)|^{k-1}\|_\infty \\
    & \times
    \|\max _r e_r^\top|\tilde A_{\alpha,\beta} \Mm (\widetilde{G(Ad})_{j-k}-\widetilde{G(Ad}^*)_{j-k})|\ind\|_{E,2q} \| |\tilde A_{\alpha,\beta} \Mm \tilde G(\cdot)|^{s-k}\|_\infty 
    \\
    &
    \leq \max_{i} \| \eps_{0,i}\|_{E,2q} \sum_{k=1}^{\min(s,j)} C_A^2/C_\lambda C_\lambda^{s} \delta_{2q}(\Ad_{G\ind},j-k) ,
\end{align*}
due to $\sum_{u=1}^p |\alpha_u|+|\beta_u|\leq 1$, $\|\max_u \max_r e_r^\top G_u(Ad_j)-G_u(Ad_j^*)\ind\|_{E,2q}\leq \delta_2q(Ad_G,j)$,
 and $\||\tilde A_{\alpha,\beta} \Mm \tilde G(\cdot)|^j\|_\infty \leq C_A C_\lambda^j$. Hence,

\begin{align*}
    \sum_{j=1}^\infty & \delta_q(\{\max_r X_{\cdot;r}\},j)
\leq \max_{i} \| \eps_{0,i}\|_{E,2q} C_A \sum_{j=1}^\infty \Big[\frac{C_A}{C_\lambda}\sum_{s=1}^\infty   \sum_{k=1}^{\min(s,j)} C_\lambda^{s} \delta_{2q}(\Ad_{G\ind},j-k) \\& + C_\lambda^j\Big]=\max_{i} \| \eps_{0,i}\|_{E,2q} C_A/(1-C_\lambda)(C_A/(C_\lambda(1-C_\lambda)) \sum_{j=0}^\infty \delta_{2q}(\Ad_{G\ind},j)+1).
\end{align*}

Since the constants do not depend on $d$, this is finite and $\Xb$ is $q$-stable for any $d$. 

For $\{Y_t^{(r)}, t \in \Z\}$, we have for $s=1,\dots,p, e_{2(s-1)+1}^\top Y_t^{(r)}= X_{t-s;r}$ and $e_{2(s-1)+2}^\top Y_t^{(r)}=e_r G_s(Ad_{t-s}) X_{t-s}$. Hence, for $s=1,\dots,p,j=1,2,$ we have
\begin{align*}
    \delta_{q}&(\{\max_r  e_{2(s-1)+j}^\top Y_{t}^{(r)}, t \in \Z\},j) \leq \delta_q(\{\max_r X_{t;r}, t \in \Z\},j-s)\\
    &+\| \max_r  e_r^\top (G_s(Ad_{j-s})-G_s(Ad_{j-s}^*)) X_{t-j}^*\|_{E,q}\\
    &\leq \delta_q(\{\max_r X_{\cdot;r}\},j-s)+  \delta_q(\Ad_{G\ind},j-s) \max_i \| X_{0,i}\|_{E,2q}.
\end{align*}
Note that $\max_i\| X_{0,i}\|_{E,2q}\leq C_A/(1-C_\lambda) \max_i \| \eps_{0,i}\|_{E,2q}$. Thus, we have for $k=1,\dots,2p,$
 \begin{align*}
     \sum_{j=0}^\infty& j \delta_q(\{\max_r  e_k^\top Y_{\cdot}^{(r)}\},j)\leq 
     \max_i \| \eps_{0,i}\|_{E,2q} C_A \Big[ \sum_{j=1}^\infty j C_\lambda^j + \fraco{1-C_\lambda}\sum_{j=1}^\infty j \delta_q(\Ad_{G\ind},j) \\
     &+\sum_{j=1}^\infty j \sum_{s=1}^\infty \sum_{k=1}^{\min(j,s)} C_A C_\lambda^{s-1} \delta_{2q}(\Ad_{G\ind},j-k)\Big] \\
     \leq& \max_i \| \eps_{0,i}\|_{E,2q} \frac{C_A}{1-C_\lambda}\Big[\frac{C_\lambda}{1-C_\lambda}+C_G+\sum_{k=1}^\infty \sum_{j=0}^\infty C_A  C_\lambda^{k-1} \delta_{2q}(\Ad_{G\ind},j) (j+k) \Big] \\
\leq& \max_i \| \eps_{0,i}\|_{E,2q} \frac{C_A}{1-C_\lambda}\Big[\frac{C_\lambda}{1-C_\lambda}+C_G(1+\frac{C_A(2-C_\lambda)}{(1-C_\lambda)^2})\Big] .
 \end{align*}
\end{proof}

\begin{proof}[Proof of Theorem~\ref{thm.LNAR}]
Let $\tilde \eps_{t;r}=\eps_{t;r}-1/(n-p)\sum_{l=p+1}^n \eps_{l;r}$ and $\tilde Y_t^{(r)}=\Yr-1/(n-p)\sum_{l=p+1}^n \Yr$. Furthermore, let $\hat \Gamma=[1/(n-p)\sum_{t=p}^{n-1} \tilde Y_t^{(r)}  (\tilde Y_t^{(r)})^\top]$ and $\var(Y_1^{(r)})=\Gamma$. Note that $\tilde Y_t^{(r)}$ is a $2p$ dimensional vector with $p$ independently from the dimension $d$. 
Let further $\hat \gamma:=1/(n-p)\sum_{t=p+1}^n \tilde \eps_{t;r} \tilde Y_{t-1}^{(r)}$ and $v:=(\hat w_r - w_r).$ Then, the linear system \eqref{eq.sys.est} gives 
$$
\|v\|_2 \sqrt{2p}  | \frac{v^\top \hat \gamma}{\|v\|_1}| \geq v^\top \hat \gamma= v^\top \Gamma v + v^\top (\hat \Gamma - \Gamma ) v\geq \|v\|_2^2(\rho(\Gamma^{-1})- \frac{2p}{\|v\|_1^2} |v^\top (\hat \Gamma- \Gamma)v|).  
$$
Hence, $\|v\|_2 \leq \sqrt{2p} |v^\top \hat \gamma/\|v\|_1|/(\rho(\Gamma^{-1})-2p(v^\top(\hat \Gamma-\Gamma)v/\|v\|_1^2)$. We use the Nagaev inequality for dependent random variables, see Theorem 2 in \cite{liu2013probability}, to bound $\|\hat \gamma\|_2$ and $v^\top (\hat \Gamma - \Gamma ) v/\|v\|_1^2$. For some process $X$, let $$\nu_q(X):=\sum_{j=1}^\infty (j^{q/2-1} \delta_q(X,j)^q)^{1/q+1}\leq \sum_{j=1}^\infty j \delta_q(X,j).$$ Note that for some vector $v \in \R^{2p}$ with $\|v\|_1=1$, we have $\delta_{q}(\{v^\top Y_{\cdot}^{(r)}\},j)=\| v^\top(Y_j^{(r)}-Y_j^{(r)*})\|_{E,q}\leq \sum_{s=1}^{2p} |v_s| \delta_q(\{Y_{\cdot;s}^{(r)}\},j)\leq \max_{s=1,\dots,2p} \delta_q(\{Y_{\cdot;s}^{(r)}\},j)$. Hence, $\nu_q(\{v^\top Y_{\cdot}^{(r)}\})\leq C_{\delta Y}$. Furthermore, we have 
$\delta_q(\{\eps_{t;r} v^\top  Y_{t-1}^{(r)}, t \in \Z\},j)=\|\eps_{0;r}\|_{E,q}\delta_q(\{v^\top Y_{\cdot}^{(r)}\},j-1)$.
Hence, $\nu(\{\eps_{t;r}v^\top Y_{t-1}^{(r)}, t \in \Z\})\leq \|\eps_{0;r}\|_{E,q} C_{\delta Y}.$
Since $\|a^2-b^2\|_{E,q}\leq \|a-b\|_{E,2q}(\|a\|_{E,2q}+\|b\|_{E,2q})$ and $\max_s\|Y_{0;s}^{(r)}\|_{E,2q}\leq C_A/(1-C_\lambda)$, we have 
$\nu_q(\{(v^\top Y_{\cdot}^{(r)})^2\})\leq C_{\delta Y} 2 C_A/(1-C_\lambda)$.

We have by Theorem 2 in \cite{liu2013probability} and the remark thereafter for some vector $v \in \R^{2p}, \|v\|_1=1$ and some $y \in \R$
\begin{align*}
P&(|\frac{1}{n-p} \sum_{t=p+1}^n  \eps_{t;r} v^\top(Y_{t-1}^{(r)}-\mu_{\Yr})| \leq 
\|\eps_{0;r}\|_{E,q} C_{\delta Y}y)
\\\geq&P(|\frac{1}{n-p} \sum_{t=p+1}^n  \eps_{t;r}v^\top(Y_{t-1}^{(r)}-\mu_{\Yr})| \leq \nu( \{\eps_{t;r}v^\top Y_{t-1}^{(r)}, t \in \Z\})^{1+1/q}y) \\
\geq&1-c_q(n-p)^{-q+1} y^{-q} (1+\frac{\|\eps_{0;r}
v^\top Y_{-1}^{(r)}\|_{E,q}^q}{\nu(\{ \eps_{t;r} v^\top Y_{t-1}^{(r)}, t \in \Z\})^{q+1}})-c_q^\prime \exp(-c_q(n-p) y^2)\\
&-2\exp(-c_q (n-p) y^2 \frac{\nu(\{\eps_{t;r} v^\top Y_{t-1}^{(r)}, t \in \Z\})^{2+2/q}}{\|\eps_{0;r}v^\top  Y{-1}^{(r)}\|_{E,2}^2})\\
\geq&1-c_q(n-p)^{1-q}y^{-q}-(c_q^\prime+2) \exp(-c_q (n-p)y^2)=:C_q(n,y),
\end{align*}
where $c_q,c_q^\prime$ are constants depending on $q$ only. Similarly, we obtain $P(1/(n-p) \sum_{t=p+1}^n \eps_{t;r} \leq\|\eps_{0;r}\|_{E,q} C_{\delta Y}y)\geq C_q(n,y)$, $P(1/(n-p) \sum_{t=p+1}^n v^\top (Y_{t-1}^{(r)}-\mu_{\Yr}) \leq 
\|\eps_{0;r}\|_{E,q} C_{\delta Y}y)\geq C_q(n,y)$, and
$P(|1/(n-p) \sum_{t=p+1}^n (v^\top (Y_{t-1}^{(r)}-\mu_{\Yr}))^2-v^\top \Gamma v| \leq  
2 C_A/(1-C_\lambda) C_{\delta Y}y)\geq C_q(n,y)$. 

Let $\Omega_q(n,y)=\{\omega : \,  
|1/(n-p) \sum_{t=p+1}^n \eps_{t;r}|\leq C_{\delta Y} y \text{ and for all } v \in \R^{2p} \text{ with } \linebreak \|v\|_1=1 \text{ such that }
|1/(n-p) \sum_{t=p+1}^n v^\top (Y_t^{(r)}-\mu_{\Yr}) |\leq C_{\delta Y} y,
|1/(n-p) \sum_{t=p+1}^n \linebreak  v^\top (Y_t^{(r)}-\mu_{\Yr}) \eps_{t;r}|\leq \|\eps_{0;r}\|_{E,q} C_{\delta Y} y, \text{ and }
|1/(n-p) \sum_{t=p+1}^n (v^\top (Y_{t-1}^{(r)}-\mu_{\Yr}))^2-v^\top \Gamma v| \leq
2 C_A/(1-C_\lambda) C_{\delta Y}y\}$. We have $P(\Omega_q(n,y))\geq C_q(n,y)^4.$ Note that $|\bar{AB}-\bar A\bar B|\leq|\bar{AB}-\mu_A\mu_B|+|\mu_A(\bar B-\mu_B)|+|(\bar A-\mu_A)\mu_B|+|(\bar A-\mu_A)(\bar B-\mu_B)|$. Thus, we have for $\omega \in \Omega_q(n,y)$
\begin{align*}
    \|\hat w -w\|_2 &\leq \sqrt{2p} |v^\top \hat \gamma/\|v\|_1|/\Big(\rho(\Gamma^{-1})-2p(v^\top(\hat \Gamma-\Gamma)v/\|v\|_1^2)\Big) \\
    &\leq y\frac{\sqrt{2p} C_{\delta Y}(\|\eps_{0;r}\|_{E,q}+C_{\delta Y}y+\mu_r +\|\mu_{\Yr}\|_1)}{\rho(\Gamma^{-1})-y 2p C_{\delta Y}(2C_A/(1-C_\lambda)+2\|\mu_{\Yr}\|_1+yC_{\delta Y})}.
\end{align*}

Since $\hat \mu_r = 1/(n-p) \sum_{t=p+1}^n (\Yr_{t-1})^\top (w_r - \hat w_r)+1/(n-p)\sum_{t=p+1}^n \eps_{t;r}$, we have 
$
    |\hat\mu_r-\mu_r|\leq (\|\mu_{\Yr}\|_1+y C_{\delta Y})\|\hat w_r-w_r\|_1+y C_{\delta Y}.
$
\end{proof}

\bibliographystyle{apalike}
\bibliography{bib}


\end{document}